\documentclass[9pt,sigconf]{acmart}

\renewcommand\footnotetextcopyrightpermission[1]{} % removes footnote with conference information in first column
\usepackage{graphicx}

\usepackage[utf8]{inputenc}
\usepackage[american]{babel}
\usepackage[babel]{csquotes}

\usepackage{amscd}
\usepackage{amsmath}
\usepackage{amssymb}
\usepackage{amsthm}
\usepackage{algorithm}
\usepackage{algpseudocode}

 \usepackage{array}

\usepackage{xspace}
\usepackage{paralist}
\usepackage{xifthen}
\usepackage{url}
\usepackage{comment}
\usepackage{galois}
\usepackage{balance}

\usepackage{tikz}
\usetikzlibrary{trees,decorations,arrows,automata,shadows,positioning,plotmarks,backgrounds,shapes}
\usetikzlibrary{calc,matrix,fit,petri,decorations.markings,decorations.pathmorphing,patterns,intersections,decorations.text}

\tikzstyle{mystate}=[state,inner sep=1pt,minimum size=10pt,line width=0.3mm]
\tikzstyle{myestate}=[ellipse,inner sep=1pt,line width=0.3mm]
\tikzstyle{fstate}=[state,accepting,inner sep=2pt,minimum size=3pt,line width=0.3mm]
\tikzstyle{istate}=[state,initial,inner sep=2pt,minimum size=3pt,line width=0.3mm]
\tikzstyle{mysquare}=[inner sep=3pt,minimum size=15pt,line width=0.3mm]
\newcommand{\SFSAutomatEdge}[5]{\path[->](#1) edge[#4,line width=0.3mm] node[#5] {\ensuremath{#2}} (#3);}
\usepackage{tabularx}
\usepackage{booktabs}

\newtheorem{remark}{Remark}

 \newcommand{\TODO}[1]{\textcolor{cyan}{\textbf{TODO:} #1}}
\newcommand{\NO}[1]{\textcolor{blue}{\textbf{NO:} #1}}

\newcommand{\AKS}[1]{\textcolor{violet}{\textbf{AKS:} #1}}

\newcommand\set[1]{\{ #1 \}}
\newcommand\tuple[1]{{( #1 )}}
\newcommand\tuplel[1]{{\langle #1 \rangle}}

\newcommand{\REFlem}[1]{\text{Lem.~\ref{#1}}}
\newcommand{\REFthm}[1]{\text{Thm.~\ref{#1}}}
\newcommand{\REFdef}[1]{Def.~\ref{#1}}
\newcommand{\REFalg}[1]{Alg.~\ref{#1}}
\newcommand{\REFrem}[1]{Rem.~\ref{#1}}
\newcommand{\REFexp}[1]{Ex.~\ref{#1}}
\newcommand{\REFsec}[1]{Sec.~\ref{#1}}
\newcommand{\REFprop}[1]{Prop.~\ref{#1}}
\newcommand{\REFfig}[1]{Fig.~\ref{#1}}

\newcommand{\REFcor}[1]{Cor.~\ref{#1}}

\newcommand{\ON}[1]{\operatorname{#1}}

\def\clap#1{\hbox to 0pt{\hss#1\hss}}

\makeatletter 
\newif\ifFIRST
\newif\ifSECOND
\let\LISTOP\relax
\newcommand{\List}[4][\;]{#3#1%
	\FIRSTtrue
	\@for\i:=#2\do{%
	\ifFIRST\LISTOP{\i}\FIRSTfalse\else,\LISTOP{\i}\fi%
	}%
	#1#4%
	\let\LISTOP\relax
}
\makeatother
\newcounter{DINGLIST}
\stepcounter{DINGLIST}
\newcommand{\markD}[3][\;\;]{\text{\ding{\the\numexpr171+\theDINGLIST}\stepcounter{DINGLIST}}#1#3}

\makeatletter

\newcommand{\propNeg}{\@ifstar\propNegStar\propNegNoStar}
\newcommand{\propNegStar}[1]{\ensuremath{\left(\propNegNoStar{#1}\right)}}
\newcommand{\propNegNoStar}[2][\cdot]{\ensuremath{\neg\ifthenelse{\isempty{#2}}{#1}{#2}}}

\newcommand{\propConj}{\@ifstar\propConjStar\propConjNoStar}
\newcommand{\propConjStar}[2]{\ensuremath{\left(\propConjNoStar{#1}{#2}\right)}}
\newcommand{\propConjNoStar}[3][\cdot]{\ensuremath{\ifthenelse{\isempty{#2}}{#1}{#2}\wedge\ifthenelse{\isempty{#3}}{#1}{#3}}}

\newcommand{\propDisj}{\@ifstar\propDisjStar\propDisjNoStar}
\newcommand{\propDisjStar}[2]{\ensuremath{\left(\propDisjNoStar{#1}{#2}\right)}}
\newcommand{\propDisjNoStar}[3][\cdot]{\ensuremath{\ifthenelse{\isempty{#2}}{#1}{#2}\vee\ifthenelse{\isempty{#3}}{#1}{#3}}}

\newcommand{\propImp}{\@ifstar\propImpStar\propImpNoStar}
\newcommand{\propImpStar}[2]{\ensuremath{\left(\propImpNoStar{#1}{#2}\right)}}
\newcommand{\propImpNoStar}[3][\cdot]{\ensuremath{\ifthenelse{\isempty{#2}}{#1}{#2}\Rightarrow\ifthenelse{\isempty{#3}}{#1}{#3}}}

\newcommand{\propAequ}{\@ifstar\propAequStar\propAequNoStar}
\newcommand{\propAequStar}[2]{\ensuremath{\left(\propAequNoStar{#1}{#2}\right)}}
\newcommand{\propAequNoStar}[3][\cdot]{\ensuremath{\ifthenelse{\isempty{#2}}{#1}{#2}\Leftrightarrow\ifthenelse{\isempty{#3}}{#1}{#3}}}

\newcommand{\AllQ}{\@ifstar\AllQStar\AllQNoStar}
\newcommand{\AllQStar}[3][\;]{\ensuremath{\left(\forall #2#1.#1#3\right)}}
\newcommand{\AllQNoStar}[3][\;]{\ensuremath{\forall #2#1.#1#3}}
\newcommand{\AllQu}{\@ifstar\AllQuStar\AllQuNoStar}
\newcommand{\AllQuStar}[3][\;]{\ensuremath{\left(\forall^{\infty} #2#1.#1#3\right)}}
\newcommand{\AllQuNoStar}[3][\;]{\ensuremath{\forall^{\infty} #2#1.#1#3}}

\newcommand{\ExQ}{\@ifstar\ExQStar\ExQNoStar}
\newcommand{\ExQStar}[3][\;]{\ensuremath{\left(\exists #2#1.#1#3\right)}}
\newcommand{\ExQNoStar}[3][\;]{\ensuremath{\exists #2#1.#1#3}}

\newcommand{\NExQ}{\@ifstar\NExQStar\NExQNoStar}
\newcommand{\NExQStar}[3][\;]{\ensuremath{\left(\nexists #2#1.#1#3\right)}}
\newcommand{\NExQNoStar}[3][\;]{\ensuremath{\nexists #2#1.#1#3}}

\newcommand{\UniqueQ}{\@ifstar\UniqueQStar\UniqueQNoStar}
\newcommand{\UniqueQStar}[3][\;]{\ensuremath{\left(\exists! #2#1.#1#3\right)}}
\newcommand{\UniqueQNoStar}[3][\;]{\ensuremath{\exists! #2#1.#1#3}}

  \newlength{\SFS@HEIGHT}
  \newlength{\SFS@WIDTH}
  \newcommand{\SplitX}[2]{
	    \settoheight{\SFS@HEIGHT}{$#2$}
	    \settowidth{\SFS@WIDTH}{$#2$}
	    \mbox{\begin{tikzpicture}[baseline=(current bounding box.center)]
	    \node[] (E) at (0,0) {$#1$};
	    \node[inner sep=0pt] (F) at ($(E.south west)+(1ex,-1ex)+(3ex+.5\SFS@WIDTH,-\SFS@HEIGHT)$) {$#2$};
	    \node[] (E) at (0,0) {\phantom{$#1$}};
	    \draw[fill] ($(E.east)+(1ex,0ex)$) circle (.2ex);
	    \draw[-] ($(E.east)+(1ex,0ex)$) -- ($(E.south east)+(1ex,-0.5ex)$) -- ($(E.south west)+(1ex,-0.5ex)$) -- ($(E.south west)+(1ex,-1ex)-(0,\SFS@HEIGHT)$) -- ($(E.south west)+(2.5ex,-1ex)-(0,\SFS@HEIGHT)$);
	    \draw[fill] ($(E.south west)+(2.5ex,-1ex)-(0,\SFS@HEIGHT)$) circle (.2ex);
	    \end{tikzpicture}}}
  \newcommand{\SplitS}[2]{
	    \settoheight{\SFS@HEIGHT}{$#2$}
	    \settowidth{\SFS@WIDTH}{$#2$}
	    \mbox{\begin{tikzpicture}[baseline=(current bounding box.center)]
	    \node[] (E) at (0,0) {$#1$};
	    \node[inner sep=0pt] (F) at ($(E.south west)+(1ex,0.5ex)+(3ex+.5\SFS@WIDTH,-\SFS@HEIGHT)$) {$#2$};
	    \end{tikzpicture}}}

\newcommand{\semantics}[1]{\langle\![#1]\!\rangle}
\newcommand{\Set}[2][]{\List[#1]{#2}{\left\{}{\right\}}}
\newcommand{\VSet}[2][]{\let\LISTOP\val\List[#1]{#2}{\{}{\}}}

\newcommand{\VTuple}[2][]{\let\LISTOP\val\List[#1]{#2}{(}{)}}
\newcommand{\UNION}{\@ifstar\UNIONStar\UNIONNoStar}
\newcommand{\UNIONStar}[2]{\ensuremath{\left(\UNIONNoStar{#1}{#2}\right)}}
\newcommand{\UNIONNoStar}[2]{\ensuremath{\ifthenelse{\isempty{#1}}{\cdot}{#1}\cup\ifthenelse{\isempty{#2}}{\cdot}{#2}}}

\newcommand{\UNIOND}{\@ifstar\UNIONDStar\UNIONDNoStar}
\newcommand{\UNIONDStar}[2]{\ensuremath{\left(\UNIONDNoStar{#1}{#2}\right)}}
\newcommand{\UNIONDNoStar}[2]{\ensuremath{\ifthenelse{\isempty{#1}}{\cdot}{#1}\uplus\ifthenelse{\isempty{#2}}{\cdot}{#2}}}

\newcommand{\SETMINUS}{\@ifstar\SETMINUSStar\SETMINUSNoStar}
\newcommand{\SETMINUSStar}[2]{\ensuremath{\left(\SETMINUSNoStar{#1}{#2}\right)}}
\newcommand{\SETMINUSNoStar}[2]{\ensuremath{\ifthenelse{\isempty{#1}}{\cdot}{#1}\setminus\ifthenelse{\isempty{#2}}{\cdot}{#2}}}

\newcommand{\INTERSECT}{\@ifstar\INTERSECTStar\INTERSECTNoStar}
\newcommand{\INTERSECTStar}[2]{\ensuremath{\left(\INTERSECTNoStar{#1}{#2}\right)}}
\newcommand{\INTERSECTNoStar}[2]{\ensuremath{\ifthenelse{\isempty{#1}}{\cdot}{#1}\cap\ifthenelse{\isempty{#2}}{\cdot}{#2}}}

\newcommand{\CARTPROD}{\@ifstar\CARTPRODStar\CARTPRODNoStar}
\newcommand{\CARTPRODStar}[2]{\ensuremath{\left(\CARTPRODNoStar{#1}{#2}\right)}}
\newcommand{\CARTPRODNoStar}[2]{\ensuremath{\ifthenelse{\isempty{#1}}{\cdot}{#1}\times\ifthenelse{\isempty{#2}}{\cdot}{#2}}}

\newcommand{\FINCOUNT}{\@ifstar\FinCountStar\FinCountNoStar}
\newcommand{\FinCountStar}[1]{\ensuremath{\#(\ifthenelse{\isempty{#1}}{\cdot}{#1})}}
\newcommand{\FinCountNoStar}[1]{\ensuremath{\#\left(\ifthenelse{\isempty{#1}}{\cdot}{#1}\right)}}

\makeatother

\newcommand{\real}[1]{\ifstrempty{#1}{\mathbb{R}}{\mathbb{R}^{#1}}}
\newcommand{\Z}{\mathbb{Z}}
\newcommand{\N}{\mathbb{N}}

\newcommand{\fun}{\ensuremath{\ON{\rightarrow}}}
\newcommand{\setfun}{\ensuremath{\ON{\rightrightarrows}}}
\newcommand{\SetComp}[3][]{\{#1#2#1\mid#1#3#1\}}

\newcommand{\twoup}[1]{\ensuremath{2^{#1}}}

\newcommand{\dom}[1]{\ensuremath{\mathrm{dom}(#1)}}

\newcommand{\frr}[2]{\preccurlyeq_{#1}^{#2}}

\newcommand{\frrE}[2]{\cong_{#1}^{#2}}

\newcommand{\frrosE}[1]{\cong^{os}_{#1}}

\newcommand{\AP}{\ensuremath{\mathtt{AP}}}

\newcommand{\Enab}{\ensuremath{\ON{Enab}}}
\newcommand{\Omap}[1]{\ensuremath{\lambda\ifthenelse{\isempty{#1}}{}{(#1)}}}
\newcommand{\Omapp}[1]{\ensuremath{\lambda'\ifthenelse{\isempty{#1}}{}{(#1)}}}
\newcommand{\Omapa}[1]{\ensuremath{\widehat{\lambda}\ifthenelse{\isempty{#1}}{}{(#1)}}}
\newcommand{\Omapn}[2]{\ensuremath{\lambda_{#2}\ifthenelse{\isempty{#1}}{}{(#1)}}}

\newcommand{\Sa}{\ensuremath{\widehat{S}}}
\newcommand{\SaK}{\ensuremath{\Sa^{\mathsf{K}}}}

\newcommand{\Xa}{\ensuremath{\widehat{X}}}

\newcommand{\Xoa}{\ensuremath{\widehat{X}_0}}
\newcommand{\Xao}{\ensuremath{\widehat{X}_0}} % remove

\newcommand{\Ya}{\ensuremath{\widehat{Y}}}
\newcommand{\Fa}{\ensuremath{\widehat{F}}}
\newcommand{\Ha}{\ensuremath{\widehat{H}}}

\newcommand{\xa}{\ensuremath{\widehat{x}}}

\newcommand{\pia}{\ensuremath{\widehat{\pi}}}

\newcommand{\St}{\ensuremath{\widetilde{S}}}
\newcommand{\Xt}{\ensuremath{\widetilde{X}}}
\newcommand{\Xto}{\ensuremath{\widetilde{X}_0}}
\newcommand{\Ft}{\ensuremath{\widetilde{F}}}
\newcommand{\Ht}{\ensuremath{\widetilde{H}}}

\newcommand{\xt}{\ensuremath{\widetilde{x}}}

\newcommand{\alphat}{\ensuremath{\widetilde{\alpha}}}

\newcommand{\Se}{\ensuremath{S^\star}}
\newcommand{\Xe}{\ensuremath{X^\star}}
\newcommand{\Xeo}{\ensuremath{X^\star_0}}

\newcommand{\Fe}{\ensuremath{F^\star}}
\newcommand{\He}{\ensuremath{H^\star}}

\newcommand{\ExX}{\ensuremath{\mathtt{EXP_X}}}
\newcommand{\Cover}{\ensuremath{\mathtt{Cover}}}
\newcommand{\ExF}{\ensuremath{\mathtt{EXP_F}}}
\newcommand{\ExG}{\ensuremath{\mathtt{EXP}_{\Gamma}}}

\newcommand{\true}{\ensuremath{\mathtt{true}}}

\newcommand{\Xo}{\ensuremath{X_0}}

\renewcommand{\C}{\ensuremath{\mathcal{C}}}
\newcommand{\Ca}{\ensuremath{\widehat{\mathcal{C}}}}
\newcommand{\Co}{\ensuremath{\mathcal{C^\dagger}}}

\newcommand{\CPaths}[1]{\ensuremath{\ON{CPaths}(#1)}}
\newcommand{\Out}[1]{\ensuremath{\ON{Out}(#1)}}
\newcommand{\Obs}[1]{\ensuremath{\ON{Obs}\ifthenelse{\isempty{#1}}{}{(#1)}}}
\newcommand{\Ext}[1]{\ensuremath{\ON{Ext}\ifthenelse{\isempty{#1}}{}{(#1)}}}

\newcommand{\Prefs}[1]{\ensuremath{\ON{Prefs}(#1)}}
\newcommand{\CPrefs}[1]{\ensuremath{\ON{CPrefs}(#1)}}
\newcommand{\Last}[1]{\ensuremath{\ON{Last}\ifthenelse{\isempty{#1}}{}{(#1)}}}
\newcommand{\ELast}[1]{\ensuremath{\ON{ELast}\ifthenelse{\isempty{#1}}{}{(#1)}}}
\newcommand{\LastS}[1]{\ensuremath{\ON{LastX}\ifthenelse{\isempty{#1}}{}{(#1)}}}
\newcommand{\LastSn}[2]{\ensuremath{\ON{LastX}_{#1}\ifthenelse{\isempty{#2}}{}{(#2)}}}
\newcommand{\EHistn}[2]{\ensuremath{\ON{EHist}_{#1}\ifthenelse{\isempty{#2}}{}{(#2)}}}
\newcommand{\EHist}[1]{\ensuremath{\ON{EHist}\ifthenelse{\isempty{#1}}{}{(#1)}}}
\newcommand{\Hist}[1]{\ensuremath{\ON{Hist}\ifthenelse{\isempty{#1}}{}{(#1)}}}
\newcommand{\Histn}[2]{\ensuremath{\ON{Hist}_{#1}\ifthenelse{\isempty{#2}}{}{(#2)}}}

\newcommand{\EPrefs}[1]{\ensuremath{\ON{EPrefs}(#1)}}

\newcommand{\Reach}[1]{\ensuremath{\ON{Reach}(#1)}}

\newcommand{\WIN}{\ensuremath{\mathcal{W}}}

\newcommand{\floor}[1]{\lfloor #1 \rfloor}

\begin{document}
 
\title{On Abstraction-Based Controller Design With Output Feedback}

% DOUBLE BLIND
% \author{Rupak Majumdar\inst{1} \and Necmiye Ozay\inst{2}\and Anne-Kathrin Schmuck\inst{1}}%, ...\thanks{}}
% \institute{MPI-SWS, Germany \and University of Michigan, Ann Arbor, USA}

	\author{Rupak Majumdar}
	\affiliation{%
		\institution{MPI-SWS, Germany}
	}
		\author{Necmiye Ozay}
	\affiliation{%
		\institution{Univ. of Michigan, Ann Arbor, USA}
	}	
	\author{Anne-Kathrin Schmuck}
	\affiliation{%
		\institution{MPI-SWS, Germany}
	}

\begin{abstract}
We consider abstraction-based design of \emph{output-feedback} controllers for dynamical systems with a finite set of inputs and outputs against specifications in linear-time temporal logic. 
The usual procedure for abstraction-based controller design (ABCD) first constructs a \emph{finite-state} abstraction of the underlying dynamical system, and second, uses reactive synthesis techniques to compute an abstract \emph{state-feedback} controller on the abstraction. %Third, given that a particular relation holds between the original system and its abstraction, this abstract controller can be refined to a concrete one for the original system.
In this context, our contribution is two-fold: (I) we define a suitable relation between the original system and its abstraction which characterizes the soundness and completeness
conditions for an abstract \emph{state-feedback} controller to be refined to a concrete \emph{output-feedback} controller for the original system, 
and 
(II) we provide an algorithm to compute a \emph{sound finite-state abstraction} fulfilling this relation. %We thereby provide a new framework for sound ABCD under output feedback. 

% Our notion of \emph{sound abstractions} for ABCD under output feedback re-interprets the well known feedback-refinement relation \cite{ReissigWeberRungger_2017_FRR} over the external input/output traces of the given dynamical system. 

Our relation generalizes feedback-refinement relations from ABCD with state-feedback.
Our algorithm for constructing sound finite-state abstractions is inspired by the simultaneous reachability and bisimulation minimization algorithm of Lee and Yannakakis.
%  \cite{LeeYannakakis92}. 
We lift their idea to the computation of an observation-equivalent system and show how \emph{sound abstractions} 
can be obtained by stopping this algorithm at any point.
Additionally, our new algorithm produces a realization of the topological closure of the input/output behavior of the original system if it is finite-state realizable. 
\end{abstract}

\maketitle

\sloppy
 %!TEX root = main.tex

\section{Introduction}
\label{sec:intro}
Controller synthesis for dynamical systems against specifications in linear temporal logic
is a core problem in correct-by-construction design of cyber-physical systems.
One way to solve this problem relies on abstracting the state space to a finite-state system,
followed by algorithmic techniques from reactive synthesis to compute an abstract controller which is then refined to a concrete one for the original system \cite{GirardPolaTabuada_2010,Tabuada09,belta2017formal,ReissigWeberRungger_2017_FRR}.
Most algorithms, and certainly most state-of-the-art synthesis tools such as SCOTS \cite{SCOTS}, pFaces \cite{pFaces},
 or Mascot \cite{HsuMMS18}, implement this abstraction-based control design (ABCD) workflow
while assuming the entire state of the underlying system to be observable. 
In this paper, we relax the condition of full state observation.
We consider ABCD when the system has a finite number of observable outputs and
a controller must decide its input choice (from a finite set) based solely on the history of applied inputs and observed outputs.
Such \emph{output-feedback control} is common in control design, as the observation of the state is usually limited by
the availability and precision of the sensors.

As an example, consider the tank reactor shown in \REFfig{fig:tank}.
It has a finite number of water level sensors ($l_0,\hdots,l_5$) which indicate whether the current water 
level touches the sensor or not by returning true or false. 
Further, it can be observed (but not controlled) whether the outlet valve is open ($o=\mathtt{true}$) or closed ($o=\mathtt{false}$).
The controller can set the inlet valve open (by applying $u=+$) or closed (by applying $u=0$). 
The actual state of the system, i.e., the precise value of the water level, is not observable. 
In this example, a given input/output sequence of observed true sensor values and applied inputs 
(e.g., $\nu=\set{l_0}\set{+}\set{l_0}\set{+}\set{l_0,l_1,o}\set{0}\set{l_0,o}\set{+}\hdots$)
provides a certain knowledge about the current true state (i.e., real water level value) 
of the tank system, which might be sufficient to implement a controller ensuring the satisfaction of a specification 
over the observables.
For example,  one might want to ensure that the tank never overflows (i.e., $l_5$ never becomes true) while still containing a limited amount of water (i.e., $l_1$ is always true).
We show how finite-state abstractions of the input/output behavior of such an 
infinite state dynamical system can be constructed for the purpose of ABCD with output-feedback. 

\begin{figure}
  \begin{tikzpicture}[auto,scale=1]
   \def\vsep{0.35}
  \begin{footnotesize}

   \def\v{0} \def\vm{5} \def\va{0.3}  
   \def\h{0} \def\hm{2} \def\ha{0.5} 
   
     \foreach \x in {0,...,\vm}{
     \node at (\h-0.3,\v+\x*\va) {$l_\x$};
     \draw [dashed] (\h,\v+\x*\va) -- (\hm,\v+\x*\va);
     }

  \draw [very thick]  (\h+\ha,\vm*\va+\va) -- (\h+\ha,\v)-- (\hm+\ha,\v);
  \draw [very thick]  (\hm,\vm*\va+\va) -- (\hm,\v+0.5*\va)-- (\hm+\ha,\v+0.5*\va);
  
  \draw [very thick] (\h-1.3*\ha,\vm*\va+1.5*\va) -- (\h+2*\ha,\vm*\va+1.5*\va) -- (\h+2*\ha,\vm*\va+0.5*\va);
  \draw [very thick] (\h-1.3*\ha,\vm*\va+2*\va) -- (\h+2.3*\ha,\vm*\va+2*\va) -- (\h+2.3*\ha,\vm*\va+0.5*\va);
  
   \draw [very thick] (\h-3*\ha,\vm*\va+1.5*\va) --(\h-1.5*\ha,\vm*\va+1.5*\va);
   \draw [very thick] (\h-3*\ha,\vm*\va+2*\va) --(\h-1.5*\ha,\vm*\va+2*\va);
 \draw []  (\h-1.5*\ha,\vm*\va+2.2*\va) -- (\h-1.3*\ha,\vm*\va+1.3*\va)-- (\h-1.5*\ha,\vm*\va+1.3*\va)-- (\h-1.3*\ha,\vm*\va+2.2*\va) -- (\h-1.5*\ha,\vm*\va+2.2*\va);
 
 \node at (\h-1.4*\ha,\vm*\va+3*\va) {$u\in\{+,0\}$};

  \draw []  (\hm+\ha,\v+0.7*\va) -- (\hm+1.2*\ha,\v-0.2*\va)-- (\hm+1*\ha,\v-0.2*\va)-- (\hm+1.2*\ha,\v+0.7*\va)-- (\hm+1*\ha,\v+0.7*\va);
  \draw [very thick] (\hm+1.2*\ha,\v+0.5*\va) -- (\hm+2*\ha,\v+0.5*\va);
  \draw [very thick] (\hm+1.2*\ha,\v+0*\va) -- (\hm+2*\ha,\v+0*\va);
  
   \node at (\hm+1*\ha,\v+1.5*\va) {$o$};
%   (\h-1.3*\ha,\vm*\va+\va) --
\end{footnotesize}

         \end{tikzpicture}
 \vspace{-0.3cm}
 \caption{Tank reactor modeled as a dynamical system $S$ over an infinite bounded state space $X\subset\mathbb{R}^3$ with finite input space $U=\Set{+,0}$ and finite output space $Y\subseteq 2^{\sigma}$ denoting the set of sensors $\sigma=\set{l_0,\hdots,l_5,o}$ which are currently `true.' 
% We assume the inlet and outlet have the same flow rate.
 }\label{fig:tank}
  \vspace{-0.3cm}
\end{figure}
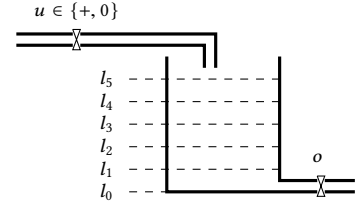

There is a rich history of output-feedback control design for continuous dynamical systems w.r.t.\
classical control objectives (such as stability or tracking) based on observer design \cite{luenberger1971introduction,shamma1999set}, with 
recent extensions to systems with finite external alphabets \cite{fan2018output} and estimator-based abstractions for control with partial-information \cite{mickelin2014synthesis,ehlers2015estimator,haesaert2015correct}.
In the context of temporal-logic control of \emph{finite-state} systems, output-feedback control
gives rise to games of incomplete information \cite{Reif,ChatterjeeDHR07,ehlers2015estimator}.
 The construction of finite-state abstractions of input/output traces for the purpose of output-feedback control is further enabled by so called $l$-complete abstractions \cite{moor1999supervisory,schmuck2014asynchronous,yang2018local,reissig2011computing}. 
 %which are inspired by \emph{behavioral systems theory} \cite{Willems}. 
Here, the underlying state dynamics of the original system are typically not assumed to be known, which is in contrast to the situation commonly handled in ABCD for dynamical systems.

% The lag of knowledge of the internal dynamics of the original systems makes the controller synthesis problem much harder and thereby typically results in less precise abstractions.

% In this work we show how the insights from $l$-complete 
% abstractions carry over to the common interpretation of ABCD and thereby partly unifying these largely separated lines of work.

% As discussed in \REFrem{rem:conectionBehavior}, these abstractions can be constructed by knowing only the external behavior of $S$.
% 
% rather limited and mostly restricted to very particular system classes \cite{}. To the best of our knowledge, the abstraction-based synthesis problem of output-feedback controllers as examplified by the above tank example, has not yet been investigated in its full generality.
In this paper we connect the above listed lines of work by building a sound ABCD framework for synthesizing output-feedback controllers for infinite-state dynamical systems with finite input and output sets.
% as has, to the best of our knowledge, not yet been considered.
In this context, our contribution is two-fold.

\textbf{(I)} We define \emph{sound abstractions} for ABCD under output feedback by relating \emph{states} of the abstract system to the \emph{external input/output traces} of the original system which directly allows to refine an abstract \emph{state-feedback} controller to an \emph{output-feedback} controller on the original system. Our relation generalizes feedback-refinement relations (FRR) \cite{ReissigWeberRungger_2017_FRR} to systems with inputs and outputs and is inspired by the framework of abstract interpretation \cite{cousot1977abstract}, which formalizes the interpretation of a given abstraction function over different system semantics. 
% \new{
% Our characterization generalizes a number of notions studied before, both in abstraction-based control \cite{HMMR00,ReissigWeberRungger_2017_FRR},
% 
% and in supervisory control of discrete event systems \cite{moor1999supervisory,schmuck2014asynchronous,yin2015uniform}.
% }

\textbf{(II)} We provide an algorithm to compute a \emph{sound finite-state abstraction} of the original infinite-state system, which we call KAM,
\emph{the Knowledge-based Abstraction with Minimization algorithm}. 
It combines two distinct ideas. First, it utilizes the \emph{forward} computation of a \emph{Knowledge-based Abstraction} (KA) typically used to solve partial observation games over \emph{finite-state} systems \cite{Reif,ChatterjeeDHR07}. 
% 
% . Starting from the initial states, KA computes all state subsets which are indistinguishable, given the past observation sequence of inputs and outputs. 
% % (see \REFfig{} \TODO{!} for an example).
% For infinite state spaces, however, KA does not terminate in general, even if there exists an exact finite-state realization of the input/output behavior of the original system, as the employed forward exploration might create more and more state-subsets even if newly generated subsets are language equivalent to already computed ones, that is, they can generate the same future sequence of outputs under the same input sequence.
% 
Second, it deploys a \emph{backward} partition refinement algorithm for bisimulation-equivalence \cite{PaigeTarjan,HenzingerMR05} to construct the language equivalence quotient of a given system. 
Neither algorithm is guaranteed to terminate for infinite-state systems, even if there exists an exact finite-state realization of the input/output behavior of the original system.
% 
%  which is not guaranteed to terminate for infite-state systems even if there exists an exact finite-state realization of the input/output behavior of the original system.
% 
% One could be tempted to circumvent this issue by first constructing the language equivalence quotient (or a stronger version, such
% as a bisimulation quotient) of the original system before applying KA and thereby \emph{Minimizing} the set of generated subsets.
% Unfortunately, this is not ideal either, as the backward exploration inherent in
% partition refinement algorithms for bisimulation-equivalence \cite{PaigeTarjan,HenzingerMR05} may split more and more states that are not distinguishable based on the input/output history, and thereby never terminates.
% 
The KAM algorithm \emph{simultaneously} executes the KA algorithm \emph{forward}, and the \emph{M}inimization of sets through refinement of 
partitions \emph{backward} and computes a finite-state realization of the topological closure of the input/output behavior of the original system if it exists. 
Further, stopping KAM after any finite number of iterations returns a \emph{sound finite-state abstraction}, even if no finite-state realization exists. 
% Our algorithm maintains an approximation of the language-equivalence relation as a \emph{covering}.
% Essentially, two language equivalent states always belong to the same sets in the covering.
% The knowledge construction occurs through forward exploration of subsets, but tags the sets based on the
% current covering.
% At the same time, a backward refinement procedure refines the current covering, but avoids splitting sets that are
% known to be indistinguishable in the forward exploration.
% The abstraction generalizes concretely reachable subsets with their representations in the covering.
% We show that the abstract state space so constructed induces a sound abstraction of the original system for abstraction-based
% output-feedback control design.
% At the same time, the algorithm terminates for many systems for which pure forward exploration or pure partition refinement does not.

% Our new KAM algorithm
The minimization part of KAM is inspired by the simultaneous reachability and bisimulation minimization algorithm of Lee and Yannakakis \cite{LeeYannakakis92}. However, as we are aiming at constructing an observation- (not bisimulation-) equivalent system, our algorithm
only applies predecessor operations and intersection with outputs, but does not take set differences. This is, indeed, in contrast to any algorithm that constructs bisimulation relations, and is crucial in implementations.
For example, one can implement KAM for linear dynamical systems by only manipulating convex polyhedra, as convexity
is maintained by both predecessor operations and intersections, but not by set difference.

To decide when KAM should terminate it must recognize when the current abstraction captures the reachable portion
of the language equivalence quotient,
which is undecidable in general.
Thus, for infinite-state systems, KAM might not realize when it should terminate, even though it
may have constructed the language equivalence quotient. 
This is also the case for the Lee-Yannakakis algorithm and the construction of $l$-complete abstractions.

We tackle the termination problem similar to the $l$-complete abstraction framework \cite{moor1999supervisory}. 
% We show that it is possible to stop KAM after any number of iterations and build a sound abstraction of the original system based on the computations done so far. 
Since KAM always constructs sound abstractions of the original system,
we can run a synthesis procedure at any point to see if an abstract controller ensuring the specification exists.
If a controller can be found, the abstraction construction can stop.
If not, the construction continues until we try again after a future iteration. 
This iterative ABCD procedure is sound and relatively complete---if a topologically closed finite-state 
abstraction that allows to construct an abstract controller for the given specification exists, our procedure will eventually find it. 

\section{Preliminaries}

% \subsection{Preliminary Definitions}

 \smallskip 
\noindent\textbf{Notation.}
We use the symbols $\N$, $\Z$,  $\real{}$, and $\real{}_{>0}$ 
to denote the sets of natural numbers, integers, reals, and positive reals, respectively. 
Given $a,b\in\real{}$ s.t.\ $a\leq b$, we denote by $[a,b]$ a closed interval and define $[a;b]=[a,b]\cap \Z$ as its integer counterpart. 
% Given \mbox{$a,b\in\real{n}$}, we denote by $a_{i}$ and $b_{i}$ their $i$-th element and write 
% $\hyint{a,b}$ for the closed hyper-interval $\real{n}\cap([a_1,b_1]\times\hdots\times[a_n,b_n])$. 
% We define the relations $<,\leq,\geq,>$ on $a,b$ component-wise.%: e.g., $a<b$ iff $\AllQ{i\in[1;n]}{a_i<b_i}$. 
% 
For a set $W$, we write $W^*$ and $W^\omega$ for the sets of finite and infinite sequences over $W$, respectively. 
% We define  $W^\infty = W^* \cup W^\omega$. 
For $w\in W^*$, we write $|w|$ for the length of $w$ and $\varepsilon$ for the empty string with $|\varepsilon|=0$; the length of $w\in W^\omega$ is $\infty$. 
We define $\dom{w} = \Set{0,\ldots, |w|-1}$ if $w\in W^*$, and $\dom{w} = \N$ if $w\in W^\omega$. 
% We denote by $\domp{w}=\SETMINUS{\dom{w}}{\Set{0}}$ the positive domain of $w$. 
For $k\in \dom{w}$ we write $w(k)$ for the $k$-th symbol of $w$
and $w|_{[0;k]}$ for the restriction of $w$ to the domain $[0;k]$. 
% If $W=A\times B$, the projection of $w\in W^\infty$ on $A$ is denoted by $w|_{A}$. 
%Furthermore, $w\sconc w'$ for $w\in W^*$ and $w'\in W^{\infty}$ denotes the concatenation of two strings. 
% The \textit{prefix relation} on strings is defined by $w\sqsubseteq w'$ if 
% ${\ExQ{w''\in W^*}{w\sconc w''=w'}}$. 
% 
Given two sets $A$ and $B$, $f:A\setfun B$ and $f:A\fun B$ denote a set-valued and ordinary map, respectively. $f$ is called \emph{strict} if $f(a)\neq\emptyset$ for all $a\in A$.  
% We identify set-valued maps with their respective binary relation over $A\times B$, 
% i.e., $(a,b)\in f$ iff $b\in f(a)$. 
The inverse mapping $f^{-1}:B\setfun A$ is defined via its respective binary relation: $f^{-1}(b)=\SetComp{a\in A}{b\in f(a)}$. 
By slightly abusing notation, we lift maps to subsets of their domain in the usual way, i.e., for a set-valued map $f:A\setfun B$ and $\alpha\subseteq A$ we have $f(\alpha)=\SetComp{b}{\ExQ{a\in\alpha}{b\in f(a)}}$, and similarly for ordinary maps.

 \smallskip 
\noindent\textbf{Systems.}
A \emph{system} $S=(X,\Xo,U,F,Y,H)$ consists of a state space $X$, 
a set of initial states $\Xo\subseteq X$, 
a \emph{finite} input space $U$, 
a strict
set-valued transition function $F:X\times U\setfun X$,
% (i.e., for all $x\in X$ and $u\in U$ we have $F(x,u)\neq\emptyset$), 
a \emph{finite} output space $Y$, and 
an output function $H:X\fun Y$. 
% By slightly abusing notation, we extend $H$ to sets of states $s\subseteq X$, i.e., $H(s)=\SetComp{y}{\ExQ{x\in s}{y=H(x)}}$. 
To simplify notation, we assume that $H$ respects $\Xo$, that is, if $H^{-1}(y)\cap \Xo\neq\emptyset$ we have $H^{-1}(y)\subseteq \Xo$. 
The system $S$ is called \emph{finite state} if $X$ is finite. %, and said to have \emph{finite external alphabets} if $U$ and $Y$ are finite. 
% It is said to have \emph{full state measurement} if $Y=X$ and $H = \ON{id}$ with $\ON{id}(x)=x$ for all $x\in X$, and to be \emph{fully initialized} if $\Xo=X$. 
%For notational simplicity, we denote a system with full state measurement by the four-tuple $S=\tuple{X,\Xo,U,F}$.
% and a fully initialized system with full state measurement by the three-tuple $S=\tuple{X,U,F}$. \AKS{check if we use this simplification anywhere}

% We define $\ON{Enab}(x)\subseteq U$ as the set of inputs available at a state: $\ON{Enab}(x) = \set{u\in U \mid F(x,u) \neq \emptyset}$,
% by definition, $\ON{Enab}(x)$ is non-empty for every $x\in X$. We say that $S$ is \emph{fully input enabled} if $\ON{Enab}(x)\subseteq U$ for all $x\in X$. 
% \new{
% We extend $\ON{Enab}$ to sets of states $s\subseteq X$: $\ON{Enab}(s)=\SetComp{u}{\ExQ{x\in s}{u\in\ON{Enab}(x)}}$.
% }

\smallskip
\noindent\textbf{Trace Semantics.}
A \emph{path} of $S$ is an infinite sequence $\pi=x_0u_0x_1u_1\hdots$ such that $x_0\in X_0$ and for all $k\in\N$ we have $x_{k+1}\in F(x_k,u_k)$. 
The set of all paths over $S$ is denoted by $\ON{Paths}(S)$. 
The \emph{prefix up to $x_n$} of a path $\pi$ over $S$ is denoted by $\pi_{[0;n]}$ with length $|\pi_{[0;n]}|=n+1$ and last element $\ON{Last}(\pi_{[0;n]})=x_n$. 
The set of all such prefixes is denoted by $\ON{Prefs}(S)$. 

% For a path $\pi$ over $S$ its unique \emph{external} and \emph{output} sequences are defined as $\ON{Ext}(\pi)=y_0u_0y_1u_1\hdots$ and $\ON{Out}(\pi)=y_0y_1\hdots$, 
% respectively, such that for all $k\in\mathbb{N}$ we have $y_k=H(x_k)$. 
% The sets of all external and output sequences over $S$ are denoted by $\Ext{S}$ and $\ON{Out}(S)$, respectively. 
% Further, we use the short notations $\EPrefs{S}:=\Ext{\Prefs{S}}$ and $\OPrefs{S}:=\Out{\Prefs{S}}$. % as well as $\EPaths{S}:=\Ext{\Paths{S}}$ and $\OPaths{S}:=\Out{\Paths{S}}$. 
The unique \emph{external} sequence of a path $\pi$ of $S$
is defined as $\ON{Ext}(\pi)=y_0u_0y_1u_1\hdots$, where $y_k=H(x_k)$ for all $k\in\mathbb{N}$. 
The sets of all external sequences over $S$ are denoted by $\Ext{S}$ and we define $\EPrefs{S}:=\Ext{\Prefs{S}}$.
The set $\Ext{S}$ is called \emph{topologically closed} (or \emph{closed} for short) if for any infinite sequence $\nu=y_0u_0y_1u_1\hdots\in Y(UY)^\omega$,  
whenever $\nu_{[0;k]}\in\EPrefs{S}$ for all $k\in\N$ it holds that $\nu\in\Ext{S}$. We say that $S$ has \emph{closed external behavior} if $\Ext{S}$ is closed (see, e.g., \cite{Willems} for details).%\footnote{We refer the interested reader to \cite{Willems} for examples and an insightful discussion of systems with closed and non-closed external behavior.}.}

We lift the map $\Last{}$ to external sequences and write 
$x\in\LastSn{S}{\rho}$ if there exists $\pi\in\Prefs{S}$ s.t.\ $\rho=\Ext{\pi}$ and $x=\Last{\pi}$.
For a state $x\in X$ we define all prefixes of $S$ that reach $x$ as $\Histn{S}{x}=\SetComp{\pi\in\Prefs{S}}{\Last{\pi}=x}$ and all external sequences generated by such prefixes as $\EHistn{S}{x}=\SetComp{\rho\in\EPrefs{S}}{x\in\LastSn{S}{\rho}}$. 
% A system $S$ is called \emph{observation induced} if $\LastSn{S}{\EHistn{S}{x}}=\set{x}$ for all $x\in X$. 
If the system $S$ we are referring to is clear from the context we omit the subscript $S$ from the maps $\LastS{}$ and $\EHist{}$. % to simplify notation.

% % The infinite trace semantics of the system consists of the set $\ON{Ext}(S)$ and the 
% maximal finite trace semantics of the system consists of the set $\EPrefs{S}$.
% \AKS{Do we need this anywhere?}
% \RM{not really; just abs int boilerplate. can remove}

% \emph{state-history-preserving system}
% $S^* = (\Prefs{S}, \Xo, U, F^*, Y, H^*)$, where 
% $F^*(\rho, u) = \set{\rho \cdot u \cdot x\mid x\in F(Last(\rho), u)}$
% and $H^*(\rho) = H(Last(\rho))$. 
% A state of the trace system represents a finite history of $S$, and the transitions extend the history by one step. 
% There is a Galois connection 
% $(2^{\Prefs{S}}, \subseteq) \galois{\alpha^*}{} (2^X, \subseteq)$ with 
% $\alpha^*(T) = \SetComp{x\in X}{T\subseteq\Hist{x}}$
% and 
% $\gamma^*(x) = \ON{Hist}(x)$.
% It is easy to see that $S \frrE{\alpha^*,\gamma^*} S^*$ by checking that conditions (A1)--(A4) in \REFdef{def:SoundAbs} are trivially satisfied. 
% 
% Similarly, one can associate a 

\smallskip 
\noindent\textbf{Control Strategies.}
We define \emph{state-feedback} and 
\emph{output-feedback control strategies} as functions $\C^\dagger: \Prefs{S}\fun U$ and $\C:\EPrefs{S}\fun U$, respectively. % s.t.\ for any $\nu\in\EPrefs{S}$ we have $\C(\nu)\in\Enab_{S^*}(\nu)$. 
We say that a path $\pi$ of $S$ is \emph{compliant} with $\C$ (resp. $\C^\dagger$) if for all 
$k\in\N$, we have $u(k)=\C(\Ext{\pi_{[0;k-1]}})$ (resp. $u(k)=\C^\dagger(\pi_{[0;k-1]})$). 
We denote the set of all paths and prefixes of $S$ compliant with $\C$ by $\CPaths{S,\C}$ and $\CPrefs{S,\C}$, respectively. 
% A state-based control strategy $\C$ is \emph{measurement-based} if for all $\rho,\rho'\in\Prefs{S}$ holds that $\Ext{\rho}=\Ext{\rho'}$ 
% implies $\C(\rho)=\C(\rho')$. 
% A controller $\C$ is \emph{memoryless} if $\Last{\rho}=\Last{\rho'}$ implies $\C(\rho)=\C(\rho')$. 
% 
% A control strategy $\C$ can equivalently be represented as a system $S_c=(\Xc,\set{\bot},\Uc,\Fc,U,\Hc)$ with $\Uc=Y$ s.t.\
% $u=\C(\pi)$ if there exists $\xc\in\Xc$ and $\sigma\in\Histn{S_c}{\xc}$ s.t.\ $\Ext{\sigma}=\pi$ and $\Hc(\xc)=u$.  
% %% \AKS{make more precise if there is time}
% Intuitively, the state of $S_c$ \enquote{remembers} the parts of the history of an external sequence seen so far which are important to fulfill the specification. 
% It is easy to see that $S_c$ provides an output-feedback controller as used in classical control theory (see e.g.\ \cite{ ReissigWeberRungger_2017_FRR}). \AKS{I feel there might be better references.}
% % A control strategy is said to have \emph{finite memory} if its associated system $S_c$ is finite state. 
% % \NO{CHECK: not sure being finite memory comes up again}
% 
% We occasionally use \emph{state-feedback control strategies} defined as a function 
% $\C^\dagger: \Prefs{S}\fun U$. 
% %s.t.\ for any $\pi\in\Prefs{S}$ we have $\C^\dagger(\pi)\subseteq\Enab_S{\Last{\pi}}$. 
% Then the above restrictions and notions for an output-feedback controller $\C$ can equivalently be defined for $\C^\dagger$ by interpreting $H$ as the identity map on $S$.  
% 
We further use $\Ext{S,\C}$ and $\EPrefs{S,\C}$ to denote the sets $\Ext{\ON{CPaths}(S,\C)}$ and $\Ext{\ON{CPrefs}(S,\C)}$, respectively. For a state-feedback controller $\C^\dagger$ all sets are defined analogously.
It should be noted that by defining \emph{compliance} of a controller $\C$ with a system $S$ over the set of path prefixes, the set $\Ext{S,\C}$ is topologically closed if $\Ext{S}$ is.

\smallskip 
\noindent\textbf{Control Problem.}
We consider $\omega$-regular specifications over 
a finite set of atomic input and output propositions $\AP_I$ and $\AP_O$. 
%% which are interpreted as sets of outputs through an \emph{observation map} $\Omap{}: \AP \rightarrow 2^Y\setminus\set{\emptyset}$.
We omit the standard definitions of $\omega$-regular languages (see, e.g., \cite{Thomas90,Thomas95}).
To simplify notation, we assume that $U=\twoup{\AP_I}$ and $Y=\twoup{\AP_O}$. In this setting, an $\omega$-regular specification $\psi$ can be written as a language $\semantics{\psi}\subseteq Y(UY)^\omega$ of desired external sequences.
Given a system $S$ and a specification $\psi$, the \emph{output-feedback control problem}, written $\tuplel{S, \psi}$, asks to find an output-feedback control strategy
$\C$ such that $\Ext{S, \C}\subseteq\semantics{\psi}$.
We define 
$\WIN(S,\psi) = \set{\C\mid \Ext{S,\C}\subseteq \semantics{\psi}}$ as 
the set of all such output-feedback control strategies.
For a state-feedback controller $\C^\dagger$, we define analogously the set 
$\WIN^\dagger(S,\psi)$.

\section{Abstraction-Based Controller Design with Output-Feedback}

Abstraction-Based Controller Design (ABCD) is a well-known approach to solving a controller synthesis problem for a dynamical system $S$ against specifications defined by a language $\semantics{\psi}$. Here, the dynamical system $S$ is first abstracted to a finite-state system $\Sa$ and then techniques from reactive synthesis (e.g., \cite{Thomas95,MPS95}) are used to design an abstract controller for $\Sa$ ensuring $\psi$. 

% Most work in ABCD, and in particular state-of-the-art synthesis tools s.a.\ SCOTS \cite{SCOTS}, pFaces \cite{pFaces},
%  or Mascot \cite{HsuMMS18}, assume that the system $S$ is fully state measurable, which is a very strict assumption in practice. While some authors argue that ABCD can be lifted to the situation of output-feedback control, this case is, to the best of our knowledge, not yet formally analyzed. The main challenge in ABCD with output feedback is that a relation over the state spaces of $S$ and $\Sa$ cannot be used to refine the abstract controller $\Ca$ to a concrete one $\C$ as state information is not available to $\C$. This requires an observation-based relation between $S$ and $\Sa$ for refinement which will be derived in this section.
% 
In this section, we will formalize the required relation between $S$ and $\Sa$ to refine an abstract \emph{state-feedback} controller $\Ca^\dagger$ on $\Sa$ 
to an \emph{output-feedback} controller $\C$ on $S$. 
We start our formalization by providing a general definition of \emph{sound abstractions} in \REFsec{sec:soundabs} which adapts feedback refinement relations \cite{ReissigWeberRungger_2017_FRR} to systems with finite input and output sets. We show that for this definition the usual refinement of an abstract state-feedback controller to a concrete \emph{state-feedback controller} carries over from  \cite{ReissigWeberRungger_2017_FRR}. %We further discuss how our definition matches the well known notion of FRR for systems with full state measurement.
% We then move to the case of output feedback control in \REFsec{sec:soundabsOF}. 
As the main contribution of this section, we then show in \REFsec{sec:soundabsOF} that the definition of sound abstraction needs to be applied to the external trace semantics of $S$ rather than to its state transitions to allow for ABCD with \emph{output feedback control}. 
%% \RM{I omitted this comment on abs int because it seems peripheral at this point:}
% This view is inspired by the framework of abstract interpretation \cite{cousot1977abstract}.
%  which formalizes the interpretation of particular abstraction functions over different systems semantics. 
% We show, how the desired properties of ABCD, such as soundness, refinement and compositionality carry over to this more general interpretation of FRR.

\subsection{Sound Abstractions}\label{sec:soundabs}

Given two systems we define a \emph{sound abstraction} as follows.  

\begin{definition}\label{def:SoundAbs}
Let $S = (X, \Xo, U, F, Y, H)$ and $\Sa = (\Xa, \Xoa, U, \Fa, \Ya, \Ha)$ be systems. % with observation maps $\Omap{}:Y\fun P$ and $\Omapa{}:\Ya\fun P$, respectively.
Further, let $\alpha:X\setfun \Xa$ and $\gamma:\Xa\setfun X$ be two set valued functions s.t.\ $x\in\gamma(\xa)$ iff $\xa\in\alpha(x)$.
Then we call $\Sa$ a \emph{sound abstraction} of $S$, written 
%\footnote{We make $\gamma$ explicit in this definition to emphasize that ${\xa}\subseteq\alpha(\gamma(\xa))$, where equality only holds in particular cases. However, as one can always define $\gamma$ from knowing $\alpha$, we only make $\alpha$ explicit when relating two systems to simplify notation.}
$S\frr{\alpha}{\gamma} \Sa$, if
% \AKS{or$S \galois{\alpha}{}\Sa$?}
\begin{compactenum}
\item[$\ON{(A1)}$] $\alpha(\Xo) \subseteq \Xoa$, 
% \item[$\ON{(A2)}$] $\AllQ{x\in X}{\ON{Enab}_{\Sa}(\alpha(x))\subseteq\ON{Enab}_{S}(x)}$,
\item[$\ON{(A2)}$]  $\AllQ{x\in X,u\in U}{\alpha(F(x,u))\subseteq \Fa(\alpha(x),u)}$, and
% \item[$\ON{(A2)}$]  $\AllQ{x\in X,u\in \ON{Enab}_{\Sa}(\alpha(x))}{\alpha(F(x,u))\subseteq \Fa(\alpha(x),u)}$, and
\item[$\ON{(A3)}$] $\AllQ{\xa\in\Xa}{H(\gamma(\xa))\subseteq \set{\Ha(\xa)}}$.
\end{compactenum}
$\Sa$ is  a \emph{sound realization} of $S$, written $S\frrE{\alpha}{\gamma} \Sa$, if $S\frr{\alpha}{\gamma} \Sa$ and $\Sa\frr{\gamma}{\alpha} S$.
\end{definition}

As common in abstract interpretation \cite{cousot1977abstract}, we make $\gamma$ explicit in \REFdef{def:SoundAbs} to emphasize that 
% \RM{this is type incorrect: (should we say we extend $\alpha$ and $\gamma$ to sets)}
$\set{\xa}\subseteq\alpha(\gamma(\xa))$, where equality may not hold. 
However, to simplify notation, we often omit $\gamma$ and write $\frr{\alpha}{}$ and $\frrE{\alpha}{}$, as $\gamma$ is fully determined by knowing $\alpha$. Further, we write $\frr{}{}$ to indicate that there exists $\alpha$ s.t.\  $\frr{\alpha}{}$ holds.

% This is true iff equality holds for all statements in (A1)-(A4) and requires ${\xa}=\alpha(\gamma(\xa))$. 
% \AKS{Do we need to prove this?}
% \AKS{I would like to get rid of $\gamma$ and have $\alpha^{-1}$ in (A4) instead.}
% \NO{I agree that it would be nice to get rid of $\gamma$. Since $\alpha$ does not look like an abstraction function (as in def. 7.50, 7.51 in \cite{baier2008principles}), would it be too much trouble to use a different letter and maybe define relations instead of functions? (this might also help getting rid of $\gamma$.}

% If $\Omap{}=\Omapa{}=\ON{id}$ we omit it from the denotations of sound abstractions and realizations to simplify notation.

\begin{remark}\label{rem:FRR}
Sound abstractions are an adaptation of feedback refinement relations (FRR) \cite[Def.~V.2]{ReissigWeberRungger_2017_FRR} to systems with finite input and output sets in the following sense.
% \begin{compactenum}
 
 \textbf{(A1):} An FRR is defined for fully initialized systems (i.e., $\Xo=X$), where (A1) follows from the fact that an FRR must be a \emph{strict} relation.
 
 \textbf{(A2):} To simplify notation, we assume that $F$ is a strict function\footnote{See \REFrem{rem:Fstrict} in \REFsec{sec:algo_subset} for a discussion of this choice.}. This implies that all inputs are enabled in every state, i.e., $\ON{Enab}_S(x) = \set{u\in U \mid F(x,u) \neq \emptyset}=U$ for all $x\in X$. The definition of FRR makes $\ON{Enab}(x)$ explicit by replacing (A2) with the two conditions
 \begin{compactenum}
  \item[$\ON{(A2.1)}$] $\AllQ{x\in X}{\ON{Enab}_{\Sa}(\alpha(x))\subseteq\ON{Enab}_{S}(x)}$, and
  \item[$\ON{(A2.2)}$]  $\AllQ{x\in X,u\in \ON{Enab}_{\Sa}(\alpha(x))}{\alpha(F(x,u))\subseteq \Fa(\alpha(x),u)}$
 \end{compactenum}
 which coincide with (A2) if $\ON{Enab}(x)=U$.\\
  \textbf{(A3):} An FRR is defined for systems with full state observation, i.e., $Y=X$, $\Ya=\Xa$ and $\Ha=H=\ON{id}$ with $\ON{id}(x)=x$ for all $x\in X$. This renders $Y$ infinite if $X$ is infinite and does not allow the direct interpretation of an $\omega$-regular specification over $U$ and $Y$. While our condition (A3) enables the use of a common specification for both $S$ and $\Sa$ (due to their equivalent finite input/output spaces), this is not possible in \cite{ReissigWeberRungger_2017_FRR}, due to $Y$ being infinite and $Y=X\neq\Xa=\Ya$. \cite[Def.VI.2]{ReissigWeberRungger_2017_FRR} handles this by defining a different \emph{abstract specification} from the defined FRR and the specification over the original system $S$. 
\end{remark}

% For (A1) one can verify that 
% 
% 
% 
%  For notational simplicity we assume that $F$ is a strict function. This implies that all inputs are enabled in every state, i.e., $\ON{Enab}_S(x) = \set{u\in U \mid F(x,u) \neq \emptyset}=U$ for all $x\in X$. To relate \REFdef{def:SoundAbs} to the notion of FRR \cite{ReissigWeberRungger_2017_FRR} we can however make $\ON{Enab}(x)$ explicit by replacing (A2) with the two conditions
%  \begin{compactenum}
%   \item[$\ON{(A2.1)}$] $\AllQ{x\in X}{\ON{Enab}_{\Sa}(\alpha(x))\subseteq\ON{Enab}_{S}(x)}$, and
%   \item[$\ON{(A2.2)}$]  $\AllQ{x\in X,u\in \ON{Enab}_{\Sa}(\alpha(x))}{\alpha(F(x,u))\subseteq \Fa(\alpha(x),u)}$
%  \end{compactenum}
%  which coincide with (A2) if $\ON{Enab}(x)=U$.
% Then it is easy to see that (A2.1) and (A2.2) literally match the conditions for $\alpha$ to be a 
% feedback refinement relation (FRR) as definined in \cite[Def.~V.2]{ReissigWeberRungger_2017_FRR}. Note that (A.1) and (A.3) are not required in \cite{ReissigWeberRungger_2017_FRR} as they assume fully initialized ($\Xo=X$) and fully observable ($Y=X$, $H=\ON{id}$) systems therein.
%  

Observe that for a system $S$ and its sound abstraction $\Sa$, 
corresponding states in two runs $x_0u_0x_1\ldots$ and $\hat{x}_0u_0\hat{x}_1\ldots$
stay related by $\alpha$ during arbitrarily but finite executions, 
if they start at related initial states $\hat{x}_0\in\alpha(x_0)$ (A1) and the same input sequence is applied (A2). 
In this case (A3) ensures that $S$ always produces a subset of the outputs generated by $\Sa$ in every instance of the trace. 
This implies that any arbitrarily but \emph{finite} external sequence $\nu$ generated by $\xi$ is contained in $\EPrefs{\Sa}$. Therefore, any abstract controller solving a given control problem over $\Sa$ can be guaranteed to be refinable to a sound controller for $S$, if $\Sa$ has \emph{closed external behavior}. If this is not the case, spurious infinite external traces generated by this controller on $S$ which are not contained in $\Ext{\Sa}$ might violate the specification. 
Requiring $\Sa$ to have closed external behavior is not with loss of much generality in ABCD:
any finite-state system (of the form considered in this paper) has closed external behavior,
and we require $\Sa$ to be finite-state in order to apply reactive synthesis techniques for abstract controller design anyways. 
% the algorithms we propose to compute $\Sa$ for output feedback control in \REFsec{sec:algos} always return a system with closed external behavior.
The next theorem formalizes the above discussion 
%% usage of \REFdef{def:SoundAbs} 
for ABCD with state feedback. 
The proof uses the same insights as the proof of \cite[Thm.VI.3]{ReissigWeberRungger_2017_FRR} and is therefore only provided in the appendix.

% Further, we see that
% the abstraction should always choose the input (which is  also ensured to be enabled in $S$ in this instance of the trace via (A2.1)) and 
% I.e., given that $S$ and $\Sa$ are in related states (enforced by (A1)),
% applying the same input to both systems ensures that all possible executions of concrete and abstract traces stay related for one more time step (see (A2)). Further, any trace of propositions $S$ generates on such a run is contained in the proposition-trace generated by $\Sa$ (see (A3)). I.e., if $\Sa$ can be controlled to fulfill the specification, this is also true for $S$. 
% With this insight we see that the notion of sound abstractions
% This allows for ABCD with state feedback, i.e., a 
% state-based control strategy $\Ca^\dagger:\Prefs{\Sa}\rightarrow U$ designed for $\Sa$ can be refined into the state-based control 
% strategy $\C^\dagger=\Ca\circ\alpha:\Prefs{S}\rightarrow U$, as formalized by the following theorem.

%  \NO{I think the more interesting result here is that due to (A.3), this abstraction also preserves output feedback controllers (which for instance alternating simulation does not  preserve). I suggest we state this at least as a corollary. The  implications are: FRR - Alt.Sim comparison says for refining a controller for Alt. Sim, you might need memory. Here the distinction is more severe; you cannot even refine an output feedback controller for an Alt. Sim. }

\begin{theorem}\label{thm:ABCDsound}
Let $S$ and $\Sa$ be systems s.t.\ $\Sa$ has closed external behavior.
If $S\frr{\alpha}{}\Sa$ and $\Ca^\dagger\in\WIN^\dagger(\Sa,\psi)$ then $\C^\dagger=\Ca^\dagger\circ\alpha\in\WIN^\dagger(S,\psi)$.
Further, if $S$ has closed external behavior and $S\frrE{\alpha}{}\Sa$ then
% $\Ca\in\WIN(\Sa,\psi)$ iff $\C:\Ca\circ\alpha\in\WIN(S,\psi)$.
$\WIN^\dagger(S,\psi)=\emptyset$ iff $\WIN^\dagger(\Sa,\psi)=\emptyset$.
\end{theorem}
% \NO{The theorem statement has $\C^\dagger$s, the proof has $\C$s. Proof uses $S\frrE{\alpha}{\gamma}\Sa$ but we skip $\gamma$ in the statement. Better to have consistency. Also, why the last statement is not "iff"?}

%   \new{
%  \begin{remark}
%   As (A2) literally match the conditions for $\alpha$ to be an FRR (see \REFrem{rem:FRR}), \REFthm{thm:ABCDsound} essentially coincides with \cite[Thm.VI.3]{ReissigWeberRungger_2017_FRR} for fully initialized systems with full state measurement. This can be seen by noting that 
% their definition of an abstract specification maps to our condition (A3) when $H=\Ha=\ON{id}$, and reduces to the requirement
% that $\alpha$ respects the interpretation of atomic predicates: for each $p\in\AP$, if $\Omap{}(x) = p$ then $\Omap{}(x') = p$ for every $x'\in \gamma(\alpha(x))$. 
%  \end{remark}
%  }

% For the sake of conciseness, we give a direct proof of \REFthm{thm:ABCDsound} in the appendix.
% We provide a more detailed comparison of \REFdef{def:SoundAbs} and FRR in \REFsec{}.

% \smallskip 
% \noindent 
% \textbf{Abstraction Based Controller Design with Output Feedback.}

\subsection{Sound Abstractions for Output Feedback}\label{sec:soundabsOF}
%  \NO{Given a discussion of output feedback in the previous section, the way I would motivate this section is as follows. There are many methods for designing discrete state feedback controllers. Can we create an abstraction such that a state feedback controller for the abstraction can be refined to an output feedback controller for the original system?}
Now we consider the case of output feedback. Here, the only available information about the system $S$ that we can utilize for control are external prefixes $\nu\in\EPrefs{S}$. With this, however, we usually cannot uniquely determine the current state of the system, i.e., $\LastS{\nu}$ is usually a set of states and not a singleton. Further, it is well known that any state of a system $S$ possesses the Markovian property, that is, knowing the current state of the system is enough to uniquely determine all its future behaviors, which is utilized in (A2) of \REFdef{def:SoundAbs}. This is, however, not true for the output space $Y$. In general, one needs to look at the entire history seen so far, i.e., at the generated string  $\nu\in\EPrefs{S}$, to uniquely determine all future observable behaviors of this system. This intuition is captured by the so called \emph{external trace system} $\Se$ of $S$ in which a state represents a finite external history of $S$, and the 
transitions extend the external history by one step.
\begin{definition}\label{def:Sstar}
 Given a system $S = (X, \Xo, U, F, Y, H)$, its induced \emph{external trace system} is the system 
$\Se  = (\Xe , \Xeo, U, \Fe, Y, \He)$, where 
$\Xe :=\EPrefs{S}$, 
$\Xeo:=H(\Xo)$,
$\Fe(\rho, u) := \SetComp{\rho u y}{ F(\LastS{\rho},u)\cap H^{-1}(y)\neq\emptyset}$
and $\He(\rho) := \Last{\rho}$. 
\end{definition}

It should be noted that, by definition, $\Se$ has closed external behavior. We further have $\EPrefs{S}=\EPrefs{\Se}$,  $\Ext{S}\subseteq \Ext{\Se}$, and $\Ext{S}=\Ext{\Se}$ iff $S$ has closed external behavior. That is, $\Ext{\Se}$ is the \emph{behavioral closure} of $\Ext{S}$ \cite{Willems}.

To refine an abstract state-feedback controller to an output-feedback controller for the original system, one needs to relate abstract states to external prefixes of $S$. As the latter form the state space of $\Se$, such a refinement is possible if $\Sa$ is a sound abstraction of $\Se$.
More precisely, it follows from \REFthm{thm:ABCDsound} that $\Se\frr{}{} \Sa$ implies that a state-feedback control strategy $\Ca^\dagger: \Prefs{\Sa}\fun U$ for $\Sa$ can be refined into a \emph{state-feedback} control strategy $\C^{\star\dagger}: \Prefs{\Se }\fun U$ for the external trace system $\Se $ of $S$. Now recalling the definition of $\Se $'s state space $\Xe :=\EPrefs{S}$, we see that for a string $\xi_0u_0\xi_1u_1\hdots\xi_k\in\Prefs{\Se }$ we have $\xi_i=\xi_k|_{[0;i]}$ for all $i\in[0;k]$. Therefore, $\xi_k$ carries all information needed for $\C^{\star\dagger}$'s control choice.  $\C^{\star\dagger}$ can therefore be redefined into a memoryless strategy $\C^{\star}: \Xe \fun U$, which, by definition, is an output-feedback control strategy  for the original system $S$ (as $\Xe :=\EPrefs{S}$). The following corollary of \REFthm{thm:ABCDsound} summarizes this observation.
% \NO{I think a useful discussion here would be, even when $S$ is finite state $\Se $ would be infinite state. We should indicate this to the reader and give a hint on this being a conceptual intermediate step and $\Se $ is never constructed.}
% \AKS{But why is this different to the usual definition of sound abstractions? $S$ is typically infinite state, so also there }

% an output-feedback control strategy $\Ca: \Xa^*\fun U$ solving $\tuplel{\Sa,\psi}$ (which is a state-feedback control strategy for the external trace system solving $\tuplel{\Sa^*,\Omapa{},\psi}$)  can be refined to an output-feedback controller $\C=\Ca\circ\alpha: \EPrefs{S}\fun U$ (as $\Xe :=\EPrefs{S}$ and $\alpha$ relates $\Xe $ and $\Xa^*$) solving $\tuplel{S,\psi}$. This is formalized in the next theorem.

% \NO{I don't feel strongly about this but I find the phrase ``observation map induced by a specification" a bit strange (same in theorem 3.2. Can we remove ``induced by the specification $\psi$" part?}
\begin{corollary}\label{cor:ABCDsound}
Let $S$ be a system, $\Se $ its external trace system and $\Sa$ a system with closed external behavior. 
If $\Se\frr{\alpha}{}\Sa$ and $\Ca^\dagger\in\WIN^\dagger(\Sa,\psi)$ then $\C=\Ca^\dagger\circ\alpha\in\WIN(S,\psi)$.
Further, if $S$ has closed external behavior and $\Se\frrE{\alpha}{}\Sa$ then
% $\Ca\in\WIN(\Sa,\psi)$ iff $\C:\Ca\circ\alpha\in\WIN(S,\psi)$.
$\WIN(S,\psi)=\emptyset$ iff $\WIN^\dagger(\Sa,\psi)=\emptyset$.
\end{corollary}

It should be noted that $\Se$ is infinite state even when the system $S$ is finite state.
This should not worry us too much as $S$ is typically also infinite state and we cannot efficiently check \REFdef{def:SoundAbs} over $S$ either. 
The contribution of \REFcor{cor:ABCDsound} is therefore conceptual. 
It shows that the same notion of sound abstractions developed for ABCD with state-feedback control can be utilized for output-feedback when applied to the external trace semantics of $S$ captured by $\Se$. In addition, the next section shows a construction of a finite-state (and therefore closed) abstraction $\Sa$ directly from $S$ which can be proven to be a sound abstraction of $\Se$ and thereby allows to apply \REFcor{cor:ABCDsound} to obtain a sound ABCD framework for output-feedback control without explicitly computing $\Se$.

\section{Computing Abstractions}\label{sec:algos}

We now turn to the algorithmic problem of \emph{computing} system abstractions
such that designing a state-feedback controller on the abstraction allows us,
through \REFcor{cor:ABCDsound}, to construct a corresponding output-feedback controller
for the original system.
For this we assume that the original system has an infinite state space---e.g., defined by a continuous-state
dynamical system---and our goal is to compute a \emph{finite-state abstraction} on which algorithmic
techniques for state-based controller synthesis (e.g., \cite{Thomas95,MPS95}) can be applied. 

We first recall two well-known approaches to compute such \emph{finite-state abstractions} which were developed for the setting where the original system has a \emph{finite} state space, and show that they may not terminate for \emph{infinite-state systems}, even if a finite-state realization of the topological closure of its external behavior exists. Based on this insight, we provide (\REFsec{sec:algo_new}) an algorithm for abstracting \emph{infinite-state systems} which overcomes this problem.

\subsection{Knowledge-Based Abstraction}\label{sec:algo_subset}

A standard way to solve control-strategy synthesis problems over finite-state systems with partial observation \cite{Reif,ChatterjeeDHR07,yin2015uniform} 
is to use a \emph{knowledge-based}
subset construction. % to convert an output-feedback control problem into a state-feedback control problem on the abstraction.
Starting from the subsets of initial states generating the same output, the knowledge-based subset construction 
algorithm, given in \REFalg{alg:SSA}, explores all inputs to the system 
and successively generates subsets of states that are indistinguishable given the full history of applied inputs and observed outputs. 
Such subsets $\xa$ of states of the original system $S$ become the states of the knowledge-based abstraction $\SaK:=\ON{KA}(S)$. 
% As all generated state subsets are indistinguishable, they generate the same output by definition. 
% Further, they are constructed from a different state subset by applying a particular input. 
% Therefore, the resulting state subsets can be collected into a system $\Sa=(\Xa,\Xao,U,Y,\Fa,\Ha)$. 
% 
\begin{algorithm}[t]
\caption{KA:  Knowledge-Based Abstraction} \label{alg:SSA}
\begin{algorithmic}[1]
\Require $S=(X,\Xo,U,F,Y,H)$  %from $S=(X,\Xo,U,F,Y,H)$ to $\Sk=(\Xk,\Xko,U,Y,\Fk,\Hk)$ 
%%\State $\Xao \gets \SetComp{\xa \in\twoup{X}}{\ExQ{y\in H(\Xo)}{\xa=H^{-1}(y)}}$\ \label{alg:SSA:Xao}
\State $\Xao \gets \SetComp{\Xo\cap H^{-1}(y) \in\twoup{X}\setminus\set{\emptyset}}{y\in Y}$\ \label{alg:SSA:Xao}
% \State $\Xao\gets \SetComp{q\in\twoup{X}\setminus\emptyset}{\ExQ{y\in Y}{q=H^{-1}(y)}}$ \label{alg:SSA:Xao}
%\State $\Xko \gets \SetCompX{\xk\in\Xk}{\ExQ{y\in Y}{\propConj*{\xk=H^{-1}(y)}{\Xo\cap\xk\neq\emptyset}}}$
% \State $\Xa \gets  \Xao$
\State $\Xa_{\mathit{old}} \gets  \emptyset$ and $\Xa \gets  \Xao$
%\While{$\exists s_1, s_2 \in \mathcal{L}, \sigma\in \Sigma$ s.t. $s_2 \cap \ON{Post}^G(s_1,\sigma) \neq \emptyset$ and $s_2 \cap \ON{Post}^G(s_1,\sigma) \neq s_2$}
\While{$\Xa_{\mathit{old}} \neq \Xa$}\label{line:SSA:whilestart}
\State $\Xa_{\mathit{old}} \gets \Xa $
\For{$\xa \in \Xa_{\mathit{old}}, u \in U, y\in Y$}
\State $\xa' \gets F(\xa,u)\cap H^{-1}(y)$ \label{alg:SSA:xap} %$\Fk(\xk,u) \gets \Fk \cup \{ (\xk, \sigma, \xk')\}$
% \If{$\xa'\notin \Xa$ and $\xa'\neq \emptyset$}
\State $\Xa \gets \Xa\cup\{\xa'\}$ if $\xa'\neq\emptyset$
% \EndIf
\EndFor
\EndWhile
\State Define $ \xa'\in\Fa( \xa,u)$ iff there exist $y$ s.t. $\xa' = F(\xa,u)\cap H^{-1}(y)$\label{alg:SSA:Fa}
 \State Define $\Ha(\xa)=y$ iff $y\in H(\xa)$ \label{alg:SSA:Ha}
% \State $\xk'\in \Fk(\xk,u),~y\in\Hk(\xk')$ iff $\xk'=F(\xk,u)\cap H^{-1}(y)$ and $\xk'\neq\emptyset$.
\State \textbf{return} $\SaK=(\Xa,\Xao,U,Y,\Fa,\Ha)$
\end{algorithmic}
\end{algorithm}
% 
% Let $S = (X, \Xo, U, F, Y, H)$ be a system.
% Then the \emph{knowledge-based subsystem} of $S$ is defined as the system
% $\SaK = (\Xa, \Xoa, U, \Fa, Y, \Ha)$
% where 
% $\Xa = 2^X\setminus \set{\emptyset}$,
% $\Xoa = \set{ X_0 \cap H^{-1}(y)\in\Xa \mid y\in Y}$ and
% $\xa'\in \Fa(\xa, u)$ iff 
% $\xa' = F(\xa,u) \cap H^{-1}(y)$ and $\xa' \neq \emptyset$.
Note that every reachable state $\xa$ of $\SaK$ computed via \REFalg{alg:SSA} has the property that all $x\in\xa$ have the same output;
thus, we can define $\Ha(\xa)$ as the (unique) output $H(x)$ of some $x\in\xa$.%\NO{do we want last line of Alg. 1 to return $\SaK$ (K missing in the Alg.).}
% In fact, in $\SaK$, it is sufficient to restrict the state space to the \emph{reachable} subsets:
% a set $\xa\in 2^X\setminus \set{\emptyset}$ is called \emph{reachable} if there is some $\xa_0 \in \Xao$ and
% a sequence $\xa_1, \ldots, \xa_k = \xa$ such that $\xa_{i+1} \in \Fa(\xa_i, u)$ for some $u\in U$ and all $i\in \set{0,\ldots, k-1}$.
% In the following, we shall always consider the state space of $\SaK$ to be restricted to the reachable subset.

\begin{remark}\label{rem:Fstrict}
We restrict our attention to systems with strict transition function in this paper to simplify the discussion of the KA algorithm in \REFalg{alg:SSA} and KAM in \REFalg{alg:GSSA}. If not all inputs are enabled in every state, KA would need to distinguish state sets further based on the set of available inputs. This would require the controller to \enquote{observe} the status of currently enabled inputs.
The not fully input-enabled case can be implicitly handled by introducing an \emph{observable} \enquote{dummy} state
and redirecting all transitions with disabled inputs to the dummy state. This indirectly observes the status of enabled inputs and provides a system with strict transition function. Then one can conjoin the specification with the constraint that the dummy state is never visited to obtain the original control problem. We postpone a more in-depth treatment of this implicit observation of enabled inputs to future work.
\end{remark}

The next proposition formalizes the intuition that $\SaK$ is a useful abstraction for a given output-feedback control problem over $S$. With \REFprop{prop:SSA_sound} in place, it immediately follows from \REFcor{cor:ABCDsound} that one can compute an output feedback controller  $\C:=\Ca^\dagger\circ\LastSn{\SaK}{}\in\WIN(S,\psi)$ from an abstract state-feedback controller $\Ca^\dagger\in\WIN^\dagger(\SaK,\psi)$, if it exists. % and the original synthesis problem has a solution, i.e.,  $\WIN(S,\psi)\neq\emptyset$ if and only if  $\WIN^\dagger(\SaK,\psi)\neq\emptyset$.

\begin{proposition}\label{prop:SSA_sound}
 Let $S$ be a system, $\Se$ its external trace system, and $\SaK=\ON{KA}(S)$. Then, %$\SaK$ is observation induced and 
 $\Se\frrE{\alpha}{}\SaK$ with $\alpha=\LastSn{\SaK}{}$.
\end{proposition}
\begin{proof}
To simplify notation we define $\Sa:=\SaK$.\\ % and prove the claim in two steps.\\
\begin{inparaitem}[$\blacktriangleright$]
\item We first prove that $\LastSn{\Sa}{\EHistn{\Sa}{\xa}}=\set{\xa}$ for all $\xa\in\Xa$ by picking $\pia=\xa_0u_0\xa_1u_1\hdots \xa_n$ and $\pia'=\xa'_0u_0\xa'_1u_1\hdots \xa'_n$ s.t.\ 
$\Ha(\xa_k)=\Ha(\xa'_k)$ for all $k\in[0;n]$ and 
showing $\xa_n=\xa'_n$ by induction.
\begin{inparaitem}[$\triangleright$]
  \item For $k=0$ we have $\xa_0,\xa'_0\in\Xoa$. 
As $\Ha(\xa_0)=\Ha(\xa'_0)$, we have $\xa_0=\xa_0'$.
\item Now let $k\in[1;n]$ and assume $\xa_{k-1}=\xa_{k-1}'$. 
Then it follows 
that there exists $y,y'$ s.t.\ $\xa_{k}=F(\xa_{k-1},u_{k-1})\cap H^{-1}(y)$ and $\xa'_{k}=F(\xa_{k-1},u_{k-1})\cap H^{-1}(y')$. 
Again, $\Ha(\xa_k)=\Ha(\xa'_k)$ implies $y=y'$. 
Then it is easy to see that $\xa_k=\xa_k'$. %As this is also true for $k=n$, this proves the statement.
\end{inparaitem}\\
% \item Pick $\xa\in\Xa$, $u\in U$. As $H^{-1}(y)\neq\emptyset$ for all $y$, we have that $\Fa(\xa,u)\neq\emptyset$ if $F(x,u)\neq\emptyset$. As $S$ is fully enabled, we see that $\Sa$ is also fully enabled.\\
% 
\item We now show that equality holds for (A1)-(A3) from \REFdef{def:SoundAbs}:\\
 \begin{inparaitem}[$\triangleright$]
%  \item (A1): It follows from the definition of $\Xao$ % in line~\ref{alg:SSA:Xao} in \REFalg{alg:SSA} 
 %            that $\LastSn{\Sa}{H(\Xo)}=\Xao$.
 \item (A1): By definition, $\Xeo = H(\Xo)$; and by line~\ref{alg:SSA:Xao} in \REFalg{alg:SSA}, we 
 have $\LastSn{\Sa}{H(\Xo)}=\Xao$.
%   \item (A2) Trivially holds as $S$ and $\Sa$ are fully enabled.\\
  \item (A2): Let $\xa=\LastSn{\Sa}{\nu}$ and $u\in U$. Further, let $\xa'_y=F(\xa,u)\cap H^{-1}(y)$ and define $Y'=\SetComp{y\in Y}{\xa'_y\neq\emptyset}$. Now recall that $\Fe(\nu,u)=\SetComp{\nu u y}{F(\LastSn{\Sa}{\nu},u)\cap H^{-1}(y)\neq\emptyset}$. This implies $\xa'_y\in\LastSn{\Sa}{\Fe(\nu,u)}$ if $y\in Y'$. Further, as $\LastSn{\Sa}{\EHistn{\Sa}{\xa}}=\set{\xa}$ we have $\LastSn{\Sa}{\Fe(\nu,u)}=\bigcup_{y\in Y'}\set{\xa'_y}$. 
From the definition of $\Fa$, it further follows 
% from line~\ref{} in \REFalg{alg:SSA} 
that $\xa'_y\in\Fa(\xa,u)$ if $y\in Y'$ and in particular $\Fa(\xa,u)=\bigcup_{y\in Y'}\set{\xa'_y}$. Recalling that $\xa=\LastSn{\Sa}{\nu}$ this shows that 
  $\LastSn{\Sa}{\Fe(\nu,u)}=\Fa(\LastSn{\Sa}{\nu},u)$.\\
\item (A3): Observe that $\gamma=\EHistn{\Sa}{}$ for $\alpha=\LastSn{\Sa}{}$. Then $H(\gamma(\xa))=H(\EHistn{\Sa}{\xa})=H(\set{\xa})$, hence $H(\set{\xa})=\set{\Ha(\xa)}$.
 \end{inparaitem}
\end{inparaitem}
\end{proof}

%With \REFprop{prop:SSA_sound} it immediately follows from \REFthm{thm:mainresult} that $\Ca^\dagger\in\WIN^\dagger(\Sa^K,\Omap{},\psi)$ iff $\C\in\WIN(S,\psi)$.
% 
% 
\REFalg{alg:SSA} incrementally constructs $\SaK$ from $S$ by \emph{forward exploration}
from the initial states. 
As the abstract state space $\Xa\subseteq\twoup{X}$ contains subsets of $X$ it terminates if $X$ is finite. 
This case is the one most prominently discussed in existing literature, e.g., in \cite{ChatterjeeDHR07,yin2015uniform}. 
However, \REFalg{alg:SSA} might also terminate if $\Xa$ is infinite (see, e.g., the example in \REFsec{sec:KAMexp}), given that  the necessary operations (in particular \enquote{$\ON{Post}$} and \enquote{$\ON{Intersect}$}) can be implemented if state subsets are infinite.
%%  (see \REFsec{sec:ImplementInfinite} for a discussion on this problem). 
% 
% \NO{We can either remove some details from the description of \REFexp{ex:Sbi_inf} or move it to the end of the next section since talking about bisimulation quotient might not make sense at this point.}
% 
If $X$ is infinite, \REFalg{alg:SSA} might however also not terminate even if there exists a finite-state realization of $S$. 
This is shown in \REFexp{ex:inf_GK}. It is interesting to note that this might still be the case even if $X=\Xo$. 
This can be verified by checking that \REFalg{alg:SSA} does also not terminate if all states in the system $S$ depicted in \REFfig{fig:ex_inf_GK} are initial.

\begin{example}\label{ex:inf_GK} 
Consider the infinite state system $S$ in \REFfig{fig:ex_inf_GK}, with $U=\{u\}$, $Y=\{A, B\}$. 
By omitting the trivial input, the external language $\Ext{S}$ of this system is $A(B)^+(A)^\omega \mid A(B)^\omega$,
for which one can construct a finite trace equivalent system, for instance, using one of the methods discussed in the following sections. 
Yet, Alg.~\ref{alg:SSA} will separate every state labeled with $B$, leading to an infinite chain of states with observation $B$, and will therefore not terminate. 
 \end{example}

\begin{figure}[t]
   \begin{center}
  \begin{tikzpicture}[auto,scale=1]
   \def\vsep{0.35}
  \begin{footnotesize}
  \node at (-0.1,0.7) {$S:$};     
         \node (init) at (-0.3,0) {};        
          \node (init2) at (0.5,-1) {};   
%           \node [mystate, circle split,draw] (a1) at (0,0) {$s_0$ \nodepart{lower} $a$};
%           \node [mystate, circle split,draw] (b1) at (1,0) {$s_2$ \nodepart{lower} $b$};
%           \node [mystate, circle split,draw] (a2) at (1,-1) {$s_1$ \nodepart{lower} $a$};
%           \node [mystate, circle split,draw] (b2) at (2,0) {$s_3$ \nodepart{lower} $b$};
%           \node [] (b3) at (3,0) {};     
%           \node [mystate, circle split,draw] (bn) at (4,0) {$s_n$ \nodepart{lower} $b$};
% %           \node [] (bnp) at (3,0) {
%           \node (dum) at (5,0){};

          \node [mystate, draw] (a1) at (0.2,0) {$a_1$}; \node[] at ($(a1)+(0,\vsep)$) {$A$};
          \node [mystate, draw] (b1) at (1,0) {$b_1$};\node[] at ($(b1)+(0,\vsep)$) {$B$};
          \node [mystate, draw] (a2) at (1,-1) {$a_2$};\node[] at ($(a2)+(-\vsep,\vsep)$) {$A$};
          \node [mystate, draw] (b2) at (2,0) {$b_2$};\node[] at ($(b2)+(0,\vsep)$) {$B$};
          \node [] (b3) at (2.9,0) {};     
          \node [mystate, draw] (bn) at (3.8,0) {$b_n$};\node[] at ($(bn)+(0,\vsep)$) {$B$};
%           \node [] (bnp) at (3,0) {
          \node (dum) at (4.7,0){};  
          \node at (0,-2.3) {};
  
  \SFSAutomatEdge{init}{}{a1}{}{}
   \SFSAutomatEdge{init2}{}{a2}{}{}
\SFSAutomatEdge{a1}{}{b1}{}{}
\SFSAutomatEdge{b1}{}{a2}{}{}
\SFSAutomatEdge{b1}{}{b2}{bend left}{}
\SFSAutomatEdge{b2}{}{b1}{bend left}{}
\SFSAutomatEdge{b2}{}{$(b3)+(-0.4,0.1)$}{bend left}{pos=0.45}
\SFSAutomatEdge{$(b3)+(-0.4,-0.1)$}{}{b2}{bend left}{pos=0.55}
% \SFSAutomatEdge{b2}{u}{b2}{bend left}{}
%\SFSAutomatEdge{b2}{u}{bn}{}{}
\draw [dotted, line width=1pt] ($(b2)+(0.5,0)$) -- ($(bn)-(0.5,0)$);
\draw [dotted, line width=1pt] ($(bn)+(0.5,0)$) -- (dum);

\SFSAutomatEdge{bn}{}{$(b3)+(+0.4,-0.1)$}{bend left}{pos=0.4}
\SFSAutomatEdge{$(b3)+(+0.4,+0.1)$}{}{bn}{bend left}{pos=0.5}
\SFSAutomatEdge{bn}{}{$(dum)+(-0.4,0.1)$}{bend left}{pos=0.5}
\SFSAutomatEdge{$(dum)+(-0.4,-0.1)$}{}{bn}{bend left}{pos=0.4}

\SFSAutomatEdge{b2}{}{a2}{bend left}{}
\SFSAutomatEdge{bn.south}{}{a2}{bend left}{pos=0.2}
\SFSAutomatEdge{a2}{}{a2}{loop below}{}
% \path (a2) edge [loop below, thick] node {$u$} (a2);
%\path (b2) edge [loop above, thick] node {$u$} (b2);
\end{footnotesize}
\hfill
         \end{tikzpicture}
   \begin{tikzpicture}[auto,scale=1]
   \def\vsep{0.35}
   \def\hsep{0.15}
  \begin{footnotesize}
   \node at (-1.2,0.2) {$\Sa^{\textsf{K}}:$};  
  \node (init) at (-1,0) {};  
  \node [myestate, draw] (a1) at (0,0) {$\set{a_1,a_2}$}; \node[] at ($(a1.east)+(\hsep,0)$) {$A$};
  \node [myestate, draw] (a2) at (-0.4,-0.7) {$\set{a_2}$}; \node[] at ($(a2.west)+(-\hsep,0)$) {$A$};
  \node [myestate, draw] (b1) at (0.7,-0.7) {$\set{b_1}$}; \node[] at ($(b1.east)+(\hsep,0)$) {$B$};
  \node [myestate, draw] (b2) at (0.7,-1.4) {$\set{b_2}$}; \node[] at ($(b2.east)+(\hsep,0)$) {$B$};
  \node [myestate, draw] (b3) at (0.7,-2.1) {$\set{b_1,b_3}$}; \node[] at ($(b3.east)+(\hsep,0)$) {$B$};
  \node [] (dum) at (0.7,-2.8) {};
%   \node [myestate, draw] (b2) at (0.7,-2.8) {$\set{s_{n-1},s_{n+1}}$}; \node[] at ($(b2)+(0,\vsep)$) {$b$};

\SFSAutomatEdge{init}{}{a1}{}{}
\SFSAutomatEdge{a1}{}{a2}{}{}
\SFSAutomatEdge{a2}{}{a2}{loop below}{}
\SFSAutomatEdge{a1}{}{b1}{}{}
\SFSAutomatEdge{b1}{}{a2}{}{}
\SFSAutomatEdge{b1}{}{b2}{}{}
\SFSAutomatEdge{b2}{}{b3}{}{}
\SFSAutomatEdge{b2}{}{a2}{}{}
\SFSAutomatEdge{b3}{}{a2}{}{}
\SFSAutomatEdge{b3}{}{$(dum)+(0,0.3)$}{}{}
% \SFSAutomatEdge{b2}{}{$(dum)+(0.15,0.3)$}{bend left}{pos=0.45}
% \SFSAutomatEdge{$(dum)+(-0.15,0.3)$}{u}{b2}{bend left}{pos=0.55}
\draw [dotted, line width=1pt] ($(b3)+(0,-0.3)$) -- ($(dum)$);
        \end{footnotesize}
         \end{tikzpicture} 
         
      \vspace{-0.8cm}   
      
         \begin{tikzpicture}[auto,scale=1]
   \def\vsep{0.35}
   \def\hsep{0.15}
     \begin{footnotesize}
      \node at (-1.5,0.5) {$\Sa^{\textsf{bi}}:$}; 
       \node (init) at (-1.7,0) {};  
       \node (init2) at (-0.1,-0.7) {};
  \node [myestate, draw] (a1) at (-1,0) {$\set{a_1}$}; \node[] at ($(a1.east)+(0,0.8*\vsep)$) {$A$};
  \node [myestate, draw] (a2) at (0.7,-0.7) {$\set{a_2}$}; \node[] at ($(a2)+(1.2*\hsep,+\vsep)$) {$A$};
  \node [myestate, draw] (b1) at (0.7,0) {$\set{b_n}_{n\in\N}$}; \node[] at ($(b1)+(0,\vsep)$) {$B$};
%   \node [] (dum) at (0.7,-2.1) {};
%   \node [myestate, draw] (b2) at (0.7,-2.8) {$\set{s_{n-1},s_{n+1}}$}; \node[] at ($(b2)+(0,\vsep)$) {$b$};

\SFSAutomatEdge{init}{}{a1}{}{}
\SFSAutomatEdge{init2}{}{a2}{}{}
% \SFSAutomatEdge{a1}{}{a2}{}{}
\SFSAutomatEdge{a2}{}{a2}{loop right}{}
\SFSAutomatEdge{a1}{}{b1}{}{}
\SFSAutomatEdge{b1}{}{a2}{}{}
\SFSAutomatEdge{b1}{}{b1}{loop right}{}
             \end{footnotesize}
         \end{tikzpicture} 
% \hfill
%            \begin{tikzpicture}[auto,scale=1]
%    \def\vsep{0.35}
%    \def\hsep{0.15}
%      \begin{footnotesize}
%      \node at (-1,0) {$\Sa^l:$}; 
%        \node (init) at (-0.5,0) {};  
%   \node [myestate, draw] (a1) at (0,0) {A}; \node[] at ($(a1.east)+(\hsep,0)$) {$A$};
%   \node [myestate, draw] (a2) at (-0.7,-0.7) {AA}; \node[] at ($(a2)+(0,-\vsep)$) {$A$};
%   \node [myestate, draw] (b1) at (0.3,-0.7) {AB}; \node[] at ($(b1)+(+\hsep,+\vsep)$) {$B$};
% \node [myestate, draw] (b3) at (1.2,-0.7) {BA}; \node[] at ($(b3)+(0,-\vsep)$) {$A$};
%   \node [myestate, draw] (b2) at (2.1,-0.7) {BB}; \node[] at ($(b2)+(0,-\vsep)$) {$B$};  
% %   \node [] (dum) at (0.7,-2.1) {};
% %   \node [myestate, draw] (b2) at (0.7,-2.8) {$\set{s_{n-1},s_{n+1}}$}; \node[] at ($(b2)+(0,\vsep)$) {$b$};
% 
% \SFSAutomatEdge{init}{}{a1}{}{}
% \SFSAutomatEdge{a1}{}{a2}{}{}
% \SFSAutomatEdge{a2}{}{a2}{loop left}{}
% \SFSAutomatEdge{a1}{}{b1}{}{}
% \SFSAutomatEdge{b1}{}{b3}{}{}
% \SFSAutomatEdge{b1}{}{b2}{bend left}{}
% \SFSAutomatEdge{b3}{}{a2}{bend left}{}
% \SFSAutomatEdge{b2}{}{b3}{}{}
% \SFSAutomatEdge{b2}{}{b2}{loop right}{}
%              \end{footnotesize}
%          \end{tikzpicture}
\hspace{-0.1cm}
         \begin{tikzpicture}[auto,scale=1]
   \def\vsep{0.35}
   \def\hsep{0.15}
     \begin{footnotesize}
      \node at (-1,0.4) {$\ON{KA}(\Sa^{\textsf{bi}}):$}; 
       \node (init) at (-1,0) {};  
%        \node (init2) at (-1.5,-0.7) {};
  \node [myestate, draw] (a1) at (0,0) {$\set{a_1,a_2}$}; \node[] at ($(a1.east)+(\hsep,0)$) {$A$};
  \node [myestate, draw] (a2) at (-0.7,-0.7) {$\set{a_2}$}; \node[] at ($(a2)+(0,-\vsep)$) {$A$};
  \node [myestate, draw] (b1) at (0.7,-0.7) {$\set{b_n}_{n\in\N}$}; \node[] at ($(b1)+(0,-\vsep)$) {$B$};
%   \node [] (dum) at (0.7,-2.1) {};
%   \node [myestate, draw] (b2) at (0.7,-2.8) {$\set{s_{n-1},s_{n+1}}$}; \node[] at ($(b2)+(0,\vsep)$) {$b$};

\SFSAutomatEdge{init}{}{a1}{}{}
% \SFSAutomatEdge{init2}{}{a2}{}{}
\SFSAutomatEdge{a1}{}{a2}{}{}
\SFSAutomatEdge{a2}{}{a2}{loop left}{}
\SFSAutomatEdge{a1}{}{b1}{}{}
\SFSAutomatEdge{b1}{}{a2}{}{}
\SFSAutomatEdge{b1}{}{b1}{loop right}{}
             \end{footnotesize}
         \end{tikzpicture} 
\end{center}
\vspace{-0.3cm}
\caption{The system $S$ (top left) has an infinite-state knowledge abstraction $\Sa^{\textsf{K}}$ (top right) while an exact finite-state representation of $\ON{Ext}(S)$ exists, which is correctly computed by first computing the bisimilarity abstraction $\Sa^{\textsf{bi}}$ (bottom left, see \REFsec{sec:algo_bisim}) and then applying \REFalg{alg:SSA} (bottom right).}\label{fig:ex_inf_GK}
\vspace{-0.5cm}
\end{figure}
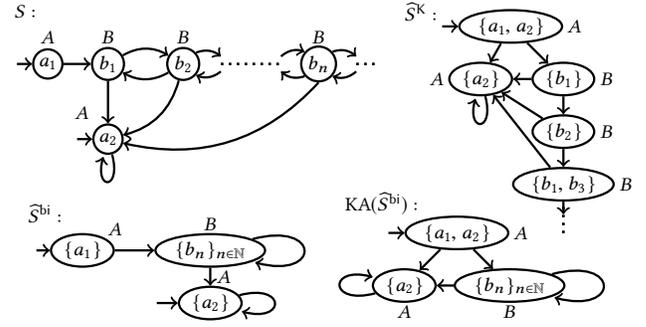

\subsection{Bisimulation Minimization}\label{sec:algo_bisim}

% Motivated by the observation that \REFalg{alg:SSA} terminates if $X$ is finite, 
% one could try to construct a finite state representation $\St$ of $S$ first, 
% and then construct a state subset abstraction $\Sa$ of $\St$. 

The knowledge-based abstraction algorithm KA computes reachable subsets going forward,
but it may fail to terminate by trying to distinguish states that are language equivalent to already computed ones, that is, 
states that generate the same future sequence of outputs under the same input sequence.
Thus, one could first compute a 
\emph{bisimulation quotient} \cite{Milner89,BouajjaniFH,HenzingerMR05} 
of the system $S$ and only then compute the knowledge-based abstraction.
It is possible that an infinite-state system has a finite bisimulation quotient;
in that case, constructing the quotient first will allow the knowledge-based abstraction to terminate (see \REFfig{fig:ex_inf_GK} (bottom) for an example).
% The correctness of the technique relies on the observation that the external language equivalence classes of a system
% can be constructed by first computing the bisimulation quotient and then performing a subset construction. 
% \AKS{Is there a reference we can cite for this claim?}
% \RM{I can't think of a direct reference. bisim preserves the language, subset makes teh autom deterministic}

For a system $S = (X, X_0, U, F, Y, H)$, a partition of the set $X$ is a set of non-empty sets of $X$, called
blocks, that are pairwise disjoint and whose union is $X$.
A partition is \emph{stable} if the following properties hold.
First, for each block $\xa$ of the partition, every state in the block
has the same output: for all $x,x'\in\xa$, we have $H(x) = H(x')$.
Second, for each pair of blocks $\xa, \xa'$ with $y'=H(x)$ for all $x\in\xa'$ and for each input $u\in U$ 
we have either $F(\xa, u)\cap H^{-1}(y') \subseteq \xa'$ or $F(\xa, u) \cap \xa' = \emptyset$.
Using the notion of a stable partition of $X$ we can define the \emph{bisimulation abstraction} $\Sa^{\mathsf{bi}} = (\Xa, \Xao, U, \Fa, Y, \Ha)$ of $S$ as follows. 
The set of abstract states $\Xa$ is the minimal stable partition of $X$.
The initial abstract states $\Xao$ are those blocks that contain some initial states from $\Xo$.
The abstract transition function is defined as $\Fa(\xa, u) = \set{\xa'\in \Xa \mid \exists x\in \xa. F(x, u) \subseteq \xa'}$.
Moreover, since every state in each block of the partition has the same output, we can uniquely
define $\Ha(\xa)$ to be the output of some state in $\xa$.

% \begin{proposition}
% Let $\alpha : X \rightarrow \Xa$ be the map from a state $x\in X$ to the block in $\Xa$ containing $x$.
% Let $\gamma: \Xa \rightarrow 2^X$ be the identity map.
% Then for any observation map $\Omap{}$, we have
% $(S, \Omap{}) \frrE{\alpha}{\gamma} (\Sa^{\mathsf{bi}}, \Omap{})$.
% Further, $(\Sa^{\mathsf{bi}})^{\mathsf{K}}$ is an observation induced abstraction of $S$.
% \end{proposition}
% \AKS{I would like to remove the proposition. It is not so interesting. It is obvious that this claim follows and we do not need it anywhere. Plus we need space.}

A partition refinement algorithm \cite{PaigeTarjan,HenzingerMR05} can be used to compute $\Sa^{\mathsf{bi}}$ from $S$.
Unlike \REFalg{alg:SSA}, this algorithm proceeds \emph{backwards} by splitting blocks based on
their predecessors, starting with the partition defined by the outputs, i.e., $\SetComp{q\in\twoup{X}\setminus\set{\emptyset}}{\ExQ{y\in Y}{q=H^{-1}(y)}}$. 
This algorithm may terminate if $X$ is infinite and 
the necessary operations are implementable over infinite state subsets. Going back to the system described in \REFexp{ex:inf_GK} we see that the bisimulation quotient $\Sa^{\mathsf{bi}}$ (depicted in \REFfig{fig:ex_inf_GK} (bottom left)) is finite, while the original system $S$ (depicted in \REFfig{fig:ex_inf_GK} (top left)) and its knowledge-based abstraction $\Sa^{\mathsf{K}}$ (depicted in \REFfig{fig:ex_inf_GK} (top right)), are infinite. Applying the KA algorithm on $\Sa^{\mathsf{bi}}$ returns the desired finite state abstraction  (depicted in \REFfig{fig:ex_inf_GK} (bottom right)) which allows for output feedback control.
However, 
if $S$ is infinite-state, the partition refinement algorithm is not guaranteed to terminate even if the knowledge-based abstraction of the original system is finite.
This is further illustrated by the example discussed in the next section, which shows that knowledge-based abstraction and bisimulation minimization are incomparable and the suggested procedure to compute $\Sa^{\mathsf{bi}}$ first, before utilizing $KA$, may not terminate.

\subsection{Illustrative Example}\label{sec:KAMexp}

\begin{figure*}[t]
 \input{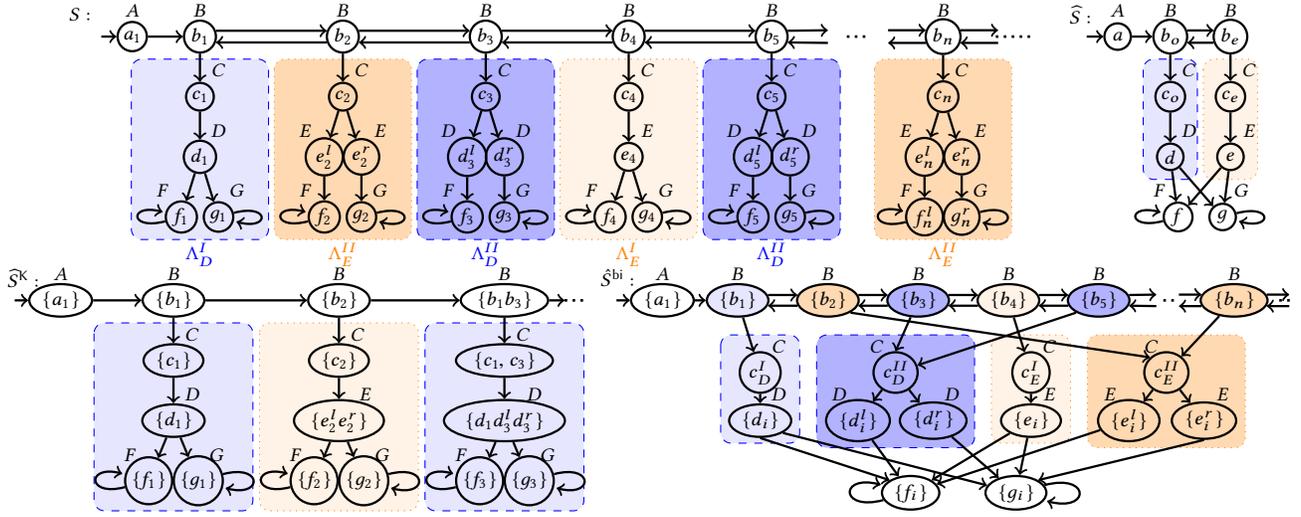}
 \vspace{-0.2cm}
 \caption{Infinite-state system $S$ (top left) discussed in \REFsec{sec:KAMexp}, its sound finite-state abstraction $\Sa$ (top right), part of its infinite-state knowledge abstraction $\SaK$ (bottom left) and its infinite bisimulation quotient $\Sa^{\mathsf{bi}}$ (bottom right). The single input $U=\{u\}$ is omitted and outputs $Y=\{A,\hdots, F\}$ are indicated next to the respective state. A state subset $\set{\alpha_i}$ denotes the set $\set{\alpha_i}_{i\in\N}$.}\label{fig:KAMexp}
  \vspace{-0.5cm}
\end{figure*}

Before explaining KAM, we introduce an illustrative example.
Consider the infinite state system $S$ depicted in \REFfig{fig:KAMexp} (top left) with $U=\{u\}$ and $Y=\{A, B, C, D, E, F\}$. It consists of one initial state $a_1$ which outputs $A$, an infinite chain of states $b_i$, $i\in\N$, all of which output $B$, and four different modules $\Lambda^I_D$ (light blue, dashed), $\Lambda^{II}_D$  (dark blue, dashed), $\Lambda^I_E$  (light orange, dotted) and $\Lambda^{II}_E$  (dark orange, dotted), attached to one $b$-state each.
System $S$ is constructed s.t.\ modules of type $D$ (resp. of type $E$) are reachable after output $B$ has occurred an \emph{odd} (resp. \emph{even}) number of times, i.e., from all states $X_B^{odd}:=\set{b_{2i+1}}_{i\in\N}$ (resp. from all states $X_B^{even}:=\set{b_{2i}}_{i\in\N}$).
However, the sequence of class $I$ and $II$ modules of the same type $i\in\set{E,D}$ is irregular, i.e., there is no $\omega$-regular expression to describe how $\Lambda^I_i$ and $\Lambda^{II}_i$ modules repeat. % for each $i\in\set{D,E}$.

By closely investigating the modules of the same $i$-type it can be observed that modules $\Lambda^I_i$ and $\Lambda^{II}_i$ for the same $i\in\set{D,E}$ are external language equivalent. Therefore, the regularity of alternating between type $D$ and type $E$ modules is enough to obtain a sound finite-state realization $\Sa$ of $S$ depicted in \REFfig{fig:KAMexp} (top right).

% I.e., the module  $\Lambda^I_i$ is a sound realization of $\Lambda^{II}_i$ but there exists no bisimulation relation between $\Lambda^I_i$ and $\Lambda^{II}_i$. 
% the external language $\ON{Ext}({S})$ of the system ${S}$ (omitting the trivial input) is
% % $(AB(C\mid D))^\omega$. 
% 
% Using this example we want to illustrate the strengths and weaknesses of the KA-algorithm (\REFsec{sec:algo_subset}) and the bisimulation algorithm (\REFsec{sec:algo_bisim}), which both do not terminate on this example. However, we will show that our new algorithm KAM, introduced in \REFsec{sec:algo_new}, successfully computes $\Sa$ depicted in \REFfig{fig:KAMexp} (top right).

\smallskip 
\noindent 
\textbf{KA-algorithm (\REFsec{sec:algo_subset}).}
The KA algorithm computes the abstract state space by combining all states with the same observable \emph{past} while going \emph{forward}. For the system $S$ in \REFfig{fig:KAMexp} (top left) it constructs state subsets as depicted in \REFfig{fig:KAMexp} (bottom left). We see that the KA algorithm 
discovers that class I modules are a sound realization of class II modules, i.e., $\SaK$ only consists of class $I$ modules s.t.\ type $D$ and type $E$ modules are reachable from states in $X_B^{odd}$ and $X_B^{even}$ respectively.
However, the KA algorithm still does not terminate on this example as it explores language equivalent states unnecessarily. I.e., by computing state subsets only going forward, it computes a new, not yet explored subset of $b$-states in every iteration. The KA-algorithm is not able to \emph{generalize} and thereby merge all states corresponding to $X_B^{odd}$ or $X_B^{even}$ due to their unique future.
% which may be much smaller than the entire state space or even finite when the original state space is infinite.
% However, it explores language equivalent states unnecessarily,
% and, by this, may not terminate.

\smallskip 
\noindent 
\textbf{Bisimulation-Quotient (\REFsec{sec:algo_bisim}).}
A partition refinement algorithm computing the bisimulation quotient of $S$ merges states with the same observable \emph{future} going \emph{backward}. For the system $S$ in \REFfig{fig:KAMexp} (top left) it immediately discovers that all states in $X_F:=\set{f_i}_{i\in\N}$ as well as $X_G:=\set{g_i}_{i\in\N}$ have the same observable future (namely $F^\omega$ and $G^\omega$, respectively). It further merges all states contained in the same $\Lambda^j_j$ module into one equivalence class (see \REFfig{fig:KAMexp} (bottom right) indicated by the four color/line patterns). However, as it proceeds backwards, it does not
take into account the reachable portion of all state subsets and thereby considers states within class I and II modules of the same type as different. This differentiates $b$ states depending on the class of modules they are connected to (indicated by the coloring of the $b$-states in \REFfig{fig:KAMexp} (bottom right)). As the 
partition refinement algorithm constructs equivalence classes going backward, it generates a distinct equivalence class for the left and right \enquote{color pattern} a $b$ state \enquote{sees}. 
As we assume that class I and II modules are irregularly sequenced, there exist infinitely many such equivalence classes and the algorithm therefore never terminates.

\smallskip 
\noindent 
\textbf{Combining both algorithms.}
For this example, running the KA algorithm first and the partition refinement algorithm second, results in the finite state abstraction $\Sa$ depicted in \REFfig{fig:KAMexp} (top right). This is, however, not practically implementable, as the KA algorithm never terminates. Further, we have shown that for \REFexp{ex:inf_GK} one needs to execute the partition refinement algorithm first, followed by the KA algorithm. One can therefore construct an example where one reachable part of the state space requires executing the KA algorithm first, while the other part requires the partition refinement algorithm to be executed first. In this case, no order would lead to the desired result. 

% Motivated by this observation, we present a new algorithm in the next subsection which interlaces the KA and the bisimulation algorithm and is able to compute the finite state abstraction $\Sa$ depicted in \REFfig{fig:KAMexp} (top right) in a finite number of iterations.

% %  is formed by irregularly concatenating an infinite number of copies of two systems $I_1$ and $II_1$. Here, by ``irregular", we mean that there is no $\omega$-regular expression to describe how $I_1$ and $II_1$ repeat. Since $I_1$ and $II_1$ are not bisimilar and they are concatenated irregularly, the bisimulation quotient of this system is infinite.
% % However, $I_1$ and $II_1$ are trace equivalent and the external language $\ON{Ext}({S})$ of the system ${S}$ (omitting the trivial input) is
% % $(AB(C\mid D))^\omega$. In this case, Alg.~\ref{alg:SSA} can capture the structure of the simple external language and 
% % terminates with an abstraction with only four states, one per output value, which essentially re-produces system $II_1$ but loops states $c$ and $d$ back to state $a$. 

% It consists of one initial node that 

%!TEX root = main.tex

\subsection{Knowledge Abstraction with Minimization}\label{sec:algo_new}

% \REFalg{alg:SSA} computes the abstract state space going forward, and hence
% it restricts its search to the reachable portion of the state space.
% % which may be much smaller than the entire state space or even finite when the original state space is infinite.
% However, it explores language equivalent states unnecessarily,
% and, by this, may not terminate.
% On the other hand, partition refinement reduces the system w.r.t.\ bisimulation (a stronger notion then language-equivalence), but since it does not
% take into account the reachable portion of all state subsets, it may 
% split states that are not distinguishable by the past observations of 
% input/output sequences of the system and, by this, may not terminate. 
% % I.e., it might split states that are always in the same equivalence class computed by the knowledge-based abstraction.

We now present the \emph{Knowledge-based Abstraction algorithm with Minimization} (KAM), 
given in \REFalg{alg:GSSA}, which interlaces the forward Knowledge-based
Abstraction (KA) with backward refinement-based Minimization (M). 
We also illustrate the algorithm using the example from \REFsec{sec:KAMexp}.
% Our algorithm maintains an approximation of the language-equivalence relation as a \emph{covering}.
% Essentially, two language equivalent states always belong to the same block in the covering.
% The knowledge construction occurs through forward exploration of subsets, but tags the computed state-sets based on the
% current covering.
% At the same time, a backward refinement procedure refines the current covering, but avoids splitting sets that are
% known to be indistinguishable in the forward exploration.
% The abstraction generalizes concretely reachable subsets with their representations in the covering. \NO{This paragraph was a bit hard for me to understand. I only understood it at a high level. ButI did not understand what the last sentence means.}
% We show that the abstract state space so constructed induces a sound abstraction of the original system for abstraction-based
% output-feedback control design.
% 

\begin{algorithm}[t]
\caption{KAM: Knowledge Abstraction and Minimization} \label{alg:GSSA}
\begin{algorithmic}[1]
 \Require $S = (X, \Xo, U, F, Y, H)$
 \State $\Cover \gets \SetComp{q\in\twoup{X}\setminus\set{\emptyset}}{\ExQ{y\in Y}{q=H^{-1}(y)}}$; \label{alg:GSSA:CoverInit}%// initialize equivalence classes
 \State $\ExG\gets\emptyset$;
 \State $\ExX \gets \SetComp{\tuplel{H(c),q,c}}{\propConj{\propConj{q\in \Cover}{c=q}}{c\cap\Xo\neq\emptyset}}$;\label{alg:GSSA:ExXInit} %// initialize explored state subsets
 \State $\ExF \gets \emptyset$; %// initialize explored transitions
%  \State $\ST \gets \ExX$ %// initialize stack for exploration
%  \State $l \gets 1$ %// initialize observation history length
\While{$\ExG\neq\SetComp{\tuplel{q,c}}{\ExQ{\nu}{\tuplel{\nu,q,c}\in\ExX}}$}\label{alg:GSSA:startwhile}
\State $\ExG\gets 
	%\ExX^\downarrow:=
	\SetComp{\tuplel{q,c}}{\ExQ{\nu}{\tuplel{\nu,q,c}\in\ExX}}$;
% \While{$\ExQ{\sigma\in\ST|_1}{|\sigma|\leq l}$}
% \State select $t=\tuplel{\sigma,q,c}\in\ST$ s.t.\ $|\sigma|\leq l$ and remove it from $\ST$;
\For{$\tuplel{\nu,q,c}\in\ExX$ s.t.\ $|\nu|$ is maximal}\label{alg:GSSA:forward-for}
\For{$u\in U$,$y\in Y$}
% \State $O_u(c)=\emptyset$;
\State $\nu'=\nu u y$;
\State $c'=F(c,u)\cap H^{-1}(y)\neq\emptyset$;\label{line:GSSA:cp}
\State $Q'=\SetComp{q'\in \Cover}{c'\subseteq q' \mbox{ and }q' \mbox{ is minimal}}$;\label{line:GSSA:Qp}
\State $\ExX\gets\ExX\cup\set{\tuplel{\nu',q',c'}\mid q'\in Q'}$;\label{line:GSSA:ExX}
% \EndIf
% \State $t' \rightarrow \ExX$;
% \If $\tuplel{q',c'}\notin\ExX$
% \State \RM{WAS:} $\ExF\gets \ExF\cup\set{(\tuplel{q,c},u,\tuplel{q',c'})}$;
\State $\ExF\gets \ExF\cup\set{(\tuplel{\nu,q,c},u,\tuplel{\nu',q',c'})\mid q'\in Q'}$;\label{line:GSSA:ExF}
\EndFor
\If{$c\subset q $}  $\textcolor{blue}{\textsc{Refine}}(\tuplel{\nu,q,c})$;\label{line:GSSA:CallRefine}
\EndIf
\EndFor
% \EndWhile
\State $\Sa\gets\textcolor{violet}{\textsc{Extract}}(\ExX,\ExF)$;\label{line:ComputeAbs:Se}
\If{%$\AllQ{\xa,u}{\Fa(\xa,u)\neq\emptyset}$ and 
\textcolor{purple}{$\ON{TermCond}()==\true$}}\label{alg:GSSA:terminate}
\textbf{return} $\Sa$;
% \Else
% \State $l=l+1$;
\EndIf
\EndWhile\label{alg:GSSA:endwhile}
\State\textbf{return} $\Sa$;
% \State \textbf{return} $\St\gets\ON{ComputeSubsetAwareAbstraction}(\Ex)$;
% \State \textbf{return} $\ST$, $\Ex$, $\QU$
\Function{\textcolor{blue}{Refine}}{$\tuplel{\nu,q,c}$}
\For{$u\in U$}
\State $\mathtt{PostQ}_u\hspace{-0.1cm}\gets\hspace{-0.1cm}\bigcup\SetComp{q'\hspace{-0.05cm}\in\hspace{-0.05cm}\Cover}{(\tuplel{\nu,q,c},u,\tuplel{\cdot,q',\cdot})\hspace{-0.05cm}\in\hspace{-0.05cm}\ExF}$;\label{line:GSSA:PostQ}
\EndFor
\State $s\gets\SetComp{x\in q}{\AllQ{u\in U}{F(x,u)\subseteq \mathtt{PostQ}_u}}$;\label{line:s}
\If{$s\subset q$} 
\State $\Cover\gets \Cover\cup \set{s}$;\label{line:s_add}
% \State $\overline{s} \gets  q\setminus s \rightarrow \Cover$;
\For{all $\tuplel{\tilde{\nu},\tilde{q},\tilde{c}}\in \ExX $ s.t.\ $\tilde{q}=q$}
\If{$\tilde{c}\hspace{-0.05cm}\subset\hspace{-0.05cm} s $} change $\tuplel{\tilde{\nu},q,\tilde{c}}$ to $\tuplel{\tilde{\nu},s,\tilde{c}}$ in $\ON{EXP}_{\Gamma,X,F}$;\label{alg:GSSA:newmin}%$\ExG$,$\ExX$,$\ExF$
\EndIf
\EndFor
% \State substitute $\tuplel{\cdot, q,c}$ by $\tuplel{\cdot,B,c}$ in $\ExX$,$\ExF$;
\For{all $(\tuplel{\tilde{\nu}',\tilde{q}',\tilde{c}'},\cdot,\tuplel{\tilde{\nu},\tilde{q},\tilde{c}})\hspace{-0.1cm}\in \hspace{-0.1cm}\ExF $ s.t.\ $\tilde{q}\hspace{-0.05cm}=\hspace{-0.05cm}s\wedge\tilde{c}'\hspace{-0.1cm}\subset \hspace{-0.05cm}\tilde{q}'$}
\State $\textcolor{blue}{\textsc{Refine}}(\tuplel{\tilde{\nu}',\tilde{q}',\tilde{c}'})$;\label{alg:GSSA:recurseRefine}
\EndFor
\EndIf
\EndFunction
\Function{\textcolor{violet}{Extract}}{\ExX,\ExF}
 \State $\Xa\gets\SetComp{q\in \twoup{X}}{ \tuplel{\cdot,q,\cdot}\in\ExX}$;
%  \State $\Xao\gets\SetComp{q\in \twoup{X}}{q\cap\Xo\neq\emptyset}$;
\State $\Xao \gets \SetComp{\Xo\cap H^{-1}(y) \in\twoup{X}\setminus\set{\emptyset}}{y\in Y}$;\ \label{alg:GSSA:Xao}
 \State $\Fa\gets\SetComp{(q,u,q')}{\tuplel{\cdot,q,\cdot},u,\tuplel{\cdot,q',\cdot})\in\ExF}$; \label{alg:GSSA:Fa}
%  \State $ q'\in\Fe( q,u)$ if  $(\tuplel{\sigma, q,\cdot},u,\tuplel{\sigma uy, q,\cdot})\in\ExF$;
%  \State $y=\He(q)$ if $\tuplel{q,\cdot}\in\ExX$;
\State $\Ha(\xa)=y$ if $y\in H(\xa)$;
 \State \textbf{return} $\Sa=(\Xa,\Xao,U,\Fa,Y,\Ha)$;
\EndFunction
\end{algorithmic}
\end{algorithm}

\begin{figure*}[t]
 \begin{tikzpicture}
 \def\vsep{0.25}
   \def\hsep{.2}
   \def\hs{2.5}
\def\vs{0.8}
   \def\ho{0.4}
 \begin{footnotesize}
%   \node at (-0.1,0.7) {$S:$};        

          \node (a1) at (0,0*\vs) {$\tuplel{A,X_A,\set{a_1}}$}; \node[blue] at ($(a1.west)+(\hsep,\vsep)$) {$t_0$};

        \node  (b1) at (0,-1*\vs) {$\tuplel{AB,X_B,\set{b_1}}$};\node[blue] at ($(b1.west)+(\hsep,\vsep)$) {$t_1$};
          \node  (b2) at (1*\hs,-2*\vs) {$\tuplel{ABB,X_B,\set{b_2}}$};\node[blue] at ($(b2.west)+(\hsep,\vsep)$) {$t_{22}$};
          \node  (b3) at (1.7*\hs,-3*\vs) {$\tuplel{ABBB,X_B,\set{b_1 b_3}}$};\node[blue] at ($(b3.west)+(\hsep,\vsep)$) {$t_{33}$};
          \node  (b4) at (2.4*\hs,-4*\vs) {$\tuplel{ABBBB,X_B,\set{b_2 b_4}}$};\node[blue] at ($(b4.west)+(\hsep,\vsep)$) {$t_{45}$};
          
          \node  (c1) at (-1*\hs,-2*\vs) {$\tuplel{ABC,X_C,\set{c_1}}$};\node[blue] at ($(c1.west)+(\hsep,\vsep)$) {$t_{21}$};
          
        \node  (c2) at (0.1*\hs,-3*\vs) {$\tuplel{ABBC,X_C,\set{c_2}}$};\node[blue] at ($(c2.west)+(\hsep,\vsep)$) {$t_{32}$};
     \node  (c3) at (1.25*\hs,-4*\vs) {$\tuplel{ABBBC,\textcolor{green!50!black}{X_C^{odd}},\set{c_1 c_3}}$};\node[blue] at ($(c3.west)+(\hsep,\vsep)$) {$t_{44}$};
          \node (d1) at (-1*\hs,-3*\vs) {$\tuplel{ABCD,X_D,\set{d_1}}$};\node[blue] at ($(d1.west)+(\hsep,\vsep)$) {$t_{31}$};
          \node (e2) at (0.1*\hs,-4*\vs) {$\tuplel{ABBCE,X_E,\set{e_2^l e_2^r}}$};\node[blue] at ($(e2.west)+(\hsep,\vsep)$) {$t_{43}$};

\node (f) at (-2*\hs,-4*\vs) {$\tuplel{ABCDF,X_F,\set{f_1}}$};\node[blue] at ($(f.west)+(\hsep,\vsep)$) {$t_{41}$};
          \node (g) at (-1*\hs,-4*\vs) {$\tuplel{ABCDG,X_G,\set{g_1}}$};\node[blue] at ($(g.west)+(\hsep,\vsep)$) {$t_{42}$};
\newcommand{\SFSAutomatEdgeL}[5]{\path[->](#1) edge[#4,line width=0.3mm] node[diamond,inner sep=0.3pt,minimum size=1pt,draw, line width=0.1pt, #5] {\ensuremath{#2}} (#3);}

\newcommand{\SFSAutomatEdgeLD}[6]{\path[->](#1) edge[#5,line width=0.3mm,red,dotted] node[state,pos=0.55,inner sep=0.3pt,minimum size=1pt,draw, line width=0.1pt, #6,red,solid] {\ensuremath{#2}}  node[pos=0.3,#6,red] {\tiny{\ensuremath{#3}}} (#4);}

\SFSAutomatEdgeL{a1.200}{1}{b1.160}{}{right}
\SFSAutomatEdgeL{b1.east}{2}{b2.160}{bend left}{above}
\SFSAutomatEdgeL{b1.west}{2}{c1.160}{bend right}{above}

\SFSAutomatEdgeL{c1}{3}{d1}{}{right}
\SFSAutomatEdgeL{b2.east}{3}{b3.160}{bend left}{right}
\SFSAutomatEdgeL{b2.west}{3}{c2.160}{bend right}{above}

\SFSAutomatEdgeL{d1.west}{4}{f}{bend right}{above}
\SFSAutomatEdgeL{d1}{4}{g}{}{right}
\SFSAutomatEdgeL{c2}{4}{e2}{}{right}
\SFSAutomatEdgeL{b3.west}{4}{c3.160}{bend right}{above}
\SFSAutomatEdgeL{b3.340}{4}{b4.160}{}{right}

\node[red] (Xcodd1) at ($(c1)+(1,0.5)$) {$X_C^{odd}$};
\SFSAutomatEdgeLD{Xcodd1.west}{3}{23}{c1.90}{bend right}{above}

\node[red] (Xbodd1) at ($(b1)+(1.2,0.5)$) {$X_B^{odd}$};
\SFSAutomatEdgeLD{Xbodd1.west}{3}{35}{b1.90}{bend right}{above}
\node[red] (Xbodd3) at ($(b3)+(1,0.5)$) {$X_B^{odd}$};
\SFSAutomatEdgeLD{Xbodd3.west}{3}{~~~~~~35/31}{b3.90}{bend right}{above}

\node[red] (Xceven2) at ($(c2)+(1,0.5)$) {$X_C^{even}$};
\SFSAutomatEdgeLD{Xceven2.west}{4}{~~23}{c2.90}{bend right}{above}

\node[red] (Xbeven2) at ($(b2)+(1,0.5)$) {$X_B^{even}$};
\SFSAutomatEdgeLD{Xbeven2.west}{4}{35}{b2.90}{bend right}{above}
\node[red] (Xbeven4) at ($(b4)+(1,0.5)$) {$X_B^{even}$};
\SFSAutomatEdgeLD{Xbeven4.west}{4}{~~~~~~35/31}{b4.90}{bend right}{above}

\draw[dotted,line width=1pt] ($(f.south)+(0,-0)$) -- ($(f.south)+(0,-0.2)$);
\draw[dotted,line width=1pt] ($(g.south)+(0,-0)$) -- ($(g.south)+(0,-0.2)$);
\draw[dotted,line width=1pt] ($(e2.south)+(0,-0)$) -- ($(e2.south)+(0,-0.2)$);
\draw[dotted,line width=1pt] ($(c3.south)+(0,-0)$) -- ($(c3.south)+(0,-0.2)$);
\draw[dotted,line width=1pt] ($(b4.south)+(0,-0)$) -- ($(b4.south)+(0,-0.2)$);

\end{footnotesize}

 \begin{footnotesize}
      \def\vsep{0.35}
   \def\hsep{0.25}
   \def\h{10.2}
   \def\vs{0.75}
    \def\hs{0.7}
   \node at (\h-1.2,-0.3*\vs) {$\hat{S}^{\sharp}_5:$};     
         \node (init) at (\h-0.8,-0.3*\vs) {};        

          \node [myestate, draw] (a1) at (\h,-0.3*\vs) {$X_A$}; \node[] at ($(a1)+(\hsep,\vsep)$) {$A$};

          \node [myestate, draw] (b1) at (\h,-1.1*\vs) {$X_B^{odd}$};\node[] at ($(b1)+(\hsep,\vsep)$) {$B$};
          \node [myestate, draw] (b2) at (\h+1*\hs,-2*\vs) {$X_B^{odd}$};\node[] at ($(b2)+(\hsep,\vsep)$) {$B$};
          
          \node [myestate, draw] (c1) at (\h-\hs,-2*\vs) {$X_C^{odd}$};\node[] at ($(c1)+(-\hsep,\vsep)$) {$C$};
          \node [myestate, draw] (c2) at (\h+\hs,-3*\vs) {$X_C^{even}$};\node[] at ($(c2)+(\hsep,\vsep)$) {$C$};
        \node [myestate, draw] (d) at (\h-\hs,-3*\vs) {$X_D$};\node[] at ($(d)+(\hsep,\vsep)$) {$D$};
          \node [myestate, draw] (e) at (\h+\hs,-4*\vs) {$X_E$};\node[] at ($(e)+(\hsep,\vsep)$) {$E$};
          
          \node [myestate, draw] (f) at (\h-1.5*\hs,-4*\vs) {$X_F$};\node[] at ($(f)+(-\hsep,\vsep)$) {$F$};
          \node [myestate, draw] (g) at (\h-0.5*\hs,-4*\vs) {$X_G$};\node[] at ($(g)+(\hsep,0.8*\vsep)$) {$G$};

\SFSAutomatEdge{init}{}{a1}{}{}  
\SFSAutomatEdge{a1}{}{b1}{}{}
\SFSAutomatEdge{b1}{}{c1}{}{}
\SFSAutomatEdge{b1}{}{b2}{}{}
\SFSAutomatEdge{b2.east}{}{b1.east}{bend right}{}
\SFSAutomatEdge{c1}{}{d}{}{}
\SFSAutomatEdge{b2}{}{c2}{}{}
\SFSAutomatEdge{c2}{}{e}{}{}
\SFSAutomatEdge{d}{}{f}{}{}
\SFSAutomatEdge{d}{}{g}{}{}
\SFSAutomatEdge{e.south west}{}{f.south east}{bend left}{}
\SFSAutomatEdge{e}{}{g}{bend right}{}
\SFSAutomatEdge{f}{}{f}{loop left}{}
\SFSAutomatEdge{g}{}{g}{loop right}{}
   
 \end{footnotesize}

 \end{tikzpicture}
\vspace{-0.5cm}
\caption{Exploration tree $\mathtt{EXP_F}$ of $S$ in \REFfig{fig:KAMexp} computed by \REFalg{alg:GSSA} (left) and the abstract system $\Sa^{\sharp}$ extracted after its 5th iteration (right). Nodes are labeled by $t_k$ (blue) for easier reference and the single input $u$ is omitted to avoid clutter. Diamond-enclosed numbers indicate the iteration in which this transition is explored. Dotted red arcs indicate cover block refinements in the iteration of the main while loop depicted by the red circled number and caused by the line of \textsc{Refine} indicated on its top right. E.g., $X_B$ of $t_1$ is refined by re-calling \textsc{Refine} in line $35$ after $X_C$ of $t_{21}$ was refined in line 23 (as $t_1$ is a predecessor of $t_{21}$). The notation ${35/31}$ in $t_{45}$ indicates that its cover block $X_B$ is refined by line $31$ after re-calling \textsc{Refine} via line $35$ on node $t_{22}$. }\label{fig:KAMwf}
\vspace{-0.5cm}
\end{figure*}
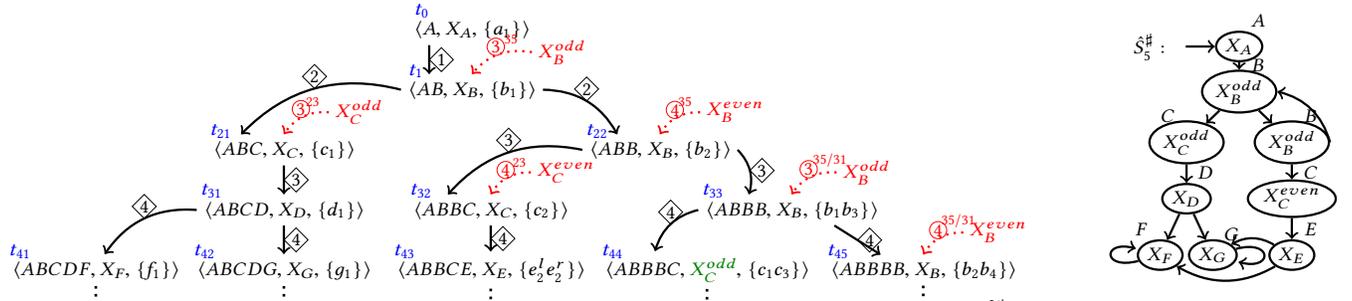

\smallskip 
\noindent 
\textbf{Algorithm Description.}
KAM generates a rooted, labeled tree
and a \emph{cover set} $\Cover \subseteq 2^X$.
The nodes of the tree are kept in $\ExX$ and the edges in $\ExF$.
The edges are labeled with inputs from $U$.
The nodes are labeled with a three-tuple $\tuplel{\nu,q,c}\in\ExX$, consisting of a
sequence $\nu$ of external events seen when reaching the current node from the root of the tree,
a \emph{block} $q\subseteq X$ in the current $\Cover$, and a subset of states $c\subseteq X$ (called a \emph{cell}).
Intuitively, a tuple $\tuplel{\nu,q,c}\in\ExX$ remembers the observed input/output sequence from the initial states (in $\nu$), 
the available knowledge about the current state (in $c$), 
and 
the current \enquote{guesses} on states which are future observation-equivalent to $c$ (in $q$).
The \emph{cells} $c$ and \emph{blocks} $q$ correspond to the data structures manipulated by the \emph{KA} and the Minimization algorithm, respectively,
and are initialized similarly:
$\Cover$ is initialized with the partition induced by $H$ on $X$ (line~\ref{alg:GSSA:CoverInit}, see \REFsec{sec:algo_bisim}), 
cells are initialized with all initial cover blocks containing an initial state (line~\ref{alg:GSSA:ExXInit}). 
%% (compare \REFalg{alg:SSA}, line~\ref{alg:SSA:Xao} and \REFalg{alg:GSSA}, line~\ref{alg:GSSA:ExXInit}). 
Note that the initialization of cells simplifies as we have assumed that $H$ respects the initial state set $\Xo$. 

\begin{example}
For the example in \REFsec{sec:KAMexp}, we see that the partition induced by $H$ on $X$ results in the initial cover set 
$\Cover=\set{X_y\mid y\in Y}$ s.t.\ $X_y$ collects all states of $S$ that generate the output $y$, 
e.g., $X_A:=\set{a_1}$ and $X_C:=\set{c_i}_{i\in\N}$. 
On the other hand, there is only one initial cell, namely $\set{a_1}$ with $H(\set{a_1})=A$. 
This results in the initialization of $\ExX$ with the tuple $\tuplel{A,X_A,\set{a_1}}$ as depicted in \REFfig{fig:KAMwf} (left).
\end{example}

The main loop of KAM (lines~\ref{alg:GSSA:startwhile}--\ref{alg:GSSA:endwhile}) 
grows the tree by iterating between a forward exploration (as in KA) and  backward refinement (as in bisimulation).
The forward exploration picks the current leaves $(\nu,q,c)$ of the tree (line~\ref{alg:GSSA:forward-for}) and
executes one step of KA to generate new \emph{cells} $c'$ for every $u\in U$ and $y\in Y$ 
(compare \REFalg{alg:SSA}, line~\ref{alg:SSA:xap} and \REFalg{alg:GSSA}, line~\ref{line:GSSA:cp}).

For each minimal block $q'$ in the current $\Cover$ set that contains $c'$, KAM adds a new node
$\tuplel{\nu',q',c'}$ to the tree (line~\ref{line:GSSA:ExX}), where $\nu'$ extends the parents event sequence with the latest
input and the last output.
The edge from the parent to the new node is labeled with the input and stored in $\ExF$ (line~\ref{line:GSSA:ExF}).

\begin{example}
The resulting exploration tree for the example in \REFsec{sec:KAMexp} is depicted in \REFfig{fig:KAMwf} (left). 
Here, the diamond-enclosed number on the edges indicates the iteration of the while loop (in line~\ref{alg:GSSA:startwhile}-\ref{alg:GSSA:endwhile} 
of \REFalg{alg:GSSA}) in which this transition and its child are added to the tree. 
When comparing \REFfig{fig:KAMwf} (left)  and the KA-abstraction $\SaK$ of this example (\REFfig{fig:KAMexp} (bottom left)), 
we see that the third component of all tuples generated by KAM coincides with the abstract states generated by KA in the same iteration (i.e., in a state with the same distance from the initial state).%\NO{I don't know how to  see the tree generated by KA in \REFfig{fig:KAMexp} (bottom left)?}
\end{example}

Having thus created all the children for a node $\tuplel{\nu, q, c}$, if $c$ is a proper subset of $q$,
the next step in KAM is to check if $q$, the current guess for the observation equivalence class for $c$,
needs to be refined.
Refinement is performed by the function \textsc{Refine} (\REFalg{alg:GSSA}, line~\ref{line:GSSA:CallRefine})
and works similarly to the bisimulation algorithm.

In contrast to the usual bisimulation algorithm, $\textsc{Refine}(\tuplel{\cdot,q,c})$ only 
splits a block $q$ based on its possible successors in the tree if this split respects $c$, 
thereby avoiding the splitting of indistinguishable states, which caused the non-termination issue discussed in \REFsec{sec:algo_bisim}. 
One can intuitively think of $s\subseteq X$ computed in line~\ref{line:s} of \REFalg{alg:GSSA} as the set of all 
states which 
% reach the same blocks $q'\in\Cover$ as any (currently indistinguishable) state in $c$ under the same input $u$. \NO{rephrase the previous sentence? I can't understand it.}
% That is, all states in $s$ 
are equivalent to $c$ in terms of their one-step observable future. 
However, in contrast to the bisimulation algorithm, KAM only adds $s$ to $\Cover$ but does not add its complement $q\setminus s$ (see line~\ref{line:s_add}). 
This is due to the fact that this 
operation might not respect the currently 
available cells and again split indistinguishable states. 
If $q\setminus s$ is indeed needed, it will be discovered by another call to \textsc{Refine}. 

Summarizing the above description, we see that \textsc{Refine} refines the $\Cover$ set based on the one-step future of the computed cell. 
Given this refinement, all previously obtained relations between cells and blocks need to be re-evaluated 
as $s\subset q$ implies that $s$ is now the minimal cover of $c$, if $c$ was previously related to $q$ in $\ExX$ (see line~\ref{alg:GSSA:newmin}). 
Thus, KAM updates its guess on the set of states possibly external language equivalent to a state in $c$. 
This, however, might imply new block splits in cell/block pairs reaching $c$, which have been checked for refinement in previous iterations of the algorithm. 
This is taken care of by the recursive call to \textsc{Refine} in line~\ref{alg:GSSA:recurseRefine}. 
Note that the recursion always moves up to the parent in the tree, and thus it eventually terminates.
% To ensure that this iterative loop on \textsc{Refine} terminates, we keep the $\nu$-component in $\ExX$, 
% which renders $\ExF$ to be a loop-free tree with finite depth. 
One can show that after the recursive call to \textsc{Refine} terminates, we always have a single minimal 
cover box $q$ for every cell $c$ computed so far. 
That is, given the relation $\alphat(c)=\SetComp{q\in\Cover}{\tuplel{c,q}\in\ExX^\downarrow}$ for 
$\ExX^\downarrow := \SetComp{\tuplel{q,c}}{\ExQ{\nu}{\tuplel{\nu,q,c}\in\ExX}}$, we 
have $|\alphat(c)|=1$ (see \REFlem{lem:singlecover} in the appendix for a formal proof).

\begin{example}
For the example in \REFsec{sec:KAMexp}, we see that for the tuple $t_0$ we have $c=q$ as $X_A=\set{a_1}$, hence, \textsc{Refine} is 
not called in the first iteration of KAM. 
In its second iteration, it computes the leaves $t_{21}$ and $t_{22}$ in the main while loop and then checks the \emph{parent node} $t_1$ for refinement. 
For this, it computes all cover cells reachable by $b_1$ (which is $\ON{PostQ}=\bigcup\set{X_B,X_C}$ and then computes all states in $q=X_B$ with the same 
reachable cover blocks, which is $s=X_B$. 
As $q=s$, no split occurs and a new iteration of the main while loop starts. 
After the computation of the leaves $t_{31}-t_{33}$ KAM checks the parent node $t_{21}$ for refinement. 
Here we obtain $\ON{PostQ}=X_D$ and $s=X_C^{odd}=\set{c_{2i+1}}_{i\in\N}$. 
As $s\subset q=X_C$ the cell $X_C^{odd}$ is added to $\Cover$.
As there is no other node in the tree with a cell component contained in $X_C^{odd}$, we only update 
the block component of $t_{21}$ (indicated by the red dotted arrow pointing to it in \REFfig{fig:KAMwf}) and schedule all its predecessors for refinement. 
Therefore, node $t_1$ is checked for refinement again. 
Given the new cover cell $X_C^{odd}$ we now obtain $\ON{PostQ}=\bigcup\set{X_B,X_C^{odd}}$ and $s=X_B^{odd}$. 
This updates the cover element of $t_1$ and $t_{33}$. 
This schedules only $t_{22}$ for refinement, as $t_0$ does not fulfill the condition that $c\subset q$. 
Now it can be observed that checking $t_{22}$ for refinement still gives $\ON{PostQ}=\bigcup\set{X_B,X_C}$ as we have not yet added the cover 
element $X_B^{even}=X_B\setminus X_B^{odd}$. 
This is due to the fact that we do not know whether this element is indeed needed and respects the constructed state subsets. 
We therefore leave node $t_{22}$ unchanged and proceed to the forth iteration of the main while loop. 
This computes the leaves $t_{41}-t_{45}$. 
It should be noted that during this computation we now have the new cover cell $X_C^{odd}$ available and KAM uses 
this smaller cover cell to correctly tack the equivalence class for $t_{44}$ (indicated in green in \REFfig{fig:KAMwf}).
Now the only interesting refinement check is on $t_{32}$ which discovers the new cover element $X_C^{even}$ and induces the further refinement 
of node $t_{22}$ introducing the cover cell $X_B^{even}$. This updates $t_{22}$ and $t_{45}$. 
Due to space constraints, we do not depict the constructed tree further. 
It should however be noted that $t_{43}$ clusters $e_2^l$ and $e_2^r$ into a single cell, as these states are not distinguishable based on 
the past observations. 
Therefore, calling \textsc{Refine} on $t_{43}$ in the next iteration of KAM will not refine the equivalence 
class $X_E$ as $\ON{PostQ}=\bigcup\set{X_F,X_G}$ and we therefore obtain $s=X_E$. 
The same happens for nodes $d_i^l$ and $d_j^r$. This prevents the non-termination issue of the bisimulation algorithm for this example.
\end{example}

% the child of $t_{44}$ will contain the cell $\set{d_1,d_{3}^l,d_{3}^r}$ (compare the subset contruction for this example depicted in \REFfig{fig:KAMexp} (bottom left). This then renders this cell future equivalent to $\set{d_1}$ (in $t_{31}$) and therefore avoids the split of nodes $d_{i}^l$ and $d_{i}^r$, which caused the non-termination of the bisimulation algorithm. The same happends for states $e_i^l$ and $e_i^r$.

After exploration and refinement, KAM extracts an abstraction $\Sa$ via the function \textsc{Extract} in line~\ref{alg:GSSA:terminate}. 
Intuitively, \textsc{Extract} projects the tree in $\ExF$ to the blocks in the current $\Cover$ set which are reachable. 
It thereby \enquote{forgets} the forward-computed cells and only retains their observation-equivalent generalizations $s$.
For the example in \REFsec{sec:KAMexp} the abstraction extracted after the fifth iteration of KAM is depicted in   
\REFfig{fig:KAMwf} (right). 
It can be observed that \REFfig{fig:KAMwf} (right) coincides with the abstraction 
$\Sa$ in \REFfig{fig:KAMexp} (top right) up to a renaming of states.

\smallskip 
\noindent 
\textbf{Termination.}
Intuitively, KAM should terminate if $\Cover$ stabilizes. 
Then, all distinguishable subsets which are observation-equivalent have been discovered, and hence, imply $\Ext{S}=\Ext{\Sa}$. 
That is, we would ideally like to have $\ON{TermCond}() ==\true$ in line~\ref{alg:GSSA:terminate} iff $\Cover$ has stabilized.
Unfortunately, even if we observe that $\Cover$ has not changed in the current iteration, we do not know if it will never change again. 
This is because KAM bases its search for cover splits on the already constructed state-subsets. 
There might be a very long input/output event sequence which only causes a subset split after a long exploration phase. 
As the state space of $S$ is infinite, we cannot check if this will ever happen. 
Interestingly, this is also true for fully initialized systems (i.e., where $X=\Xo$). 
Thus, this termination check is undecidable. 
% Formally, it can be reduced to checking language equivalence for infinite state systems, which is undecidable in general.
% For example, one can construct a pushdown automaton whose language is one of two possible regular languages, 
% but it is still undecidable to decide which one it is. 

One interesting special case where termination is decidable occurs if the KA algorithm (\REFalg{alg:SSA}) terminates (which is for example always the case if $X$ is finite). 
In this case, one can show that $\ExG=\ExX^\downarrow$ holds in the $l$-th iteration of \REFalg{alg:GSSA} iff $\Gamma=\Xa$ holds in the $l$-th iteration of \REFalg{alg:SSA} (see \REFlem{lem:TerminatingGSSAifSSA} in the appendix for a formal proof of this statement). While $\Cover$ might have stabilized earlier, we know it has surely stabilized by then.

\smallskip 
\noindent 
\textbf{Finite-State Abstractions.}
The termination condition discussed above aims on computing a sound \emph{finite-state realization} of the external behavioral closure of $S$ which might not exist. 
Indeed, for arbitrary non-linear dynamical systems there rarely ever exists an exact finite-state realization in this sense, even if their input and output sets are finite. 
Therefore, as the name suggests, abstraction-based controller synthesis is usually only aiming at computing a finite-state \emph{abstraction} 
which is accurate enough to synthesize an abstract controller for the given specification. 

In this context, it is interesting to investigate whether the system $\Sa^{\#}$ computed in line~\ref{line:ComputeAbs:Se} of \REFalg{alg:GSSA} after running the while loop in line~\ref{alg:GSSA:startwhile}-\ref{alg:GSSA:endwhile} finitely often, is indeed a sound abstraction of $S$ in the sense of \REFdef{def:SoundAbs} 
and therefore allows for abstraction based control in the sense of \REFcor{cor:ABCDsound}. 
Interestingly, this is only true if KAM has already explored all possible output events which are reachable in $S$ at least once when terminated. 
This is for example trivially satisfied if $\Xo=X$. 
Additionally, whenever $\Cover$ stabilizes after a finite number of iterations, KAM indeed computes a sound realization of $S$. 
This 
is formalized in the following theorem.

 \begin{theorem}\label{thm:soundrel:Ssharp}
 Let $S$ be a system, $\Se$ its external trace system and $\Sa^{\#}$ an abstract system extracted in line~\ref{line:ComputeAbs:Se} of $\ON{KAM}(S)$ in some iteration.
%  \NO{is $l$ times redundant here because I do not see dependence on $l$ in the rest of the theorem?}
 Further, let $Y^{\#}=\SetComp{y\in Y}{\ExQ{\tuplel{q,c}\in\ExG}{\Ha(q)=y}}$ and $\Reach{Y}=\SetComp{y\in Y}{\ExQ{\rho\in\EPrefs{S}}{y=\Last{\rho}}}$.
 If $Y^{\#}=\Reach{Y}$ it holds that
 $\Se\frr{\alpha}{}\Sa^{\#}$ with $\alpha=\LastSn{\Sa^\#}{}$.
 Further, if $\Cover$ has stabilized, we additionally have $\Se\frrE{\alpha}{}\Sa^{\#}$.
\end{theorem}

In order to prove \REFthm{thm:soundrel:Ssharp}, we first prove \REFprop{prop:SaKbsimSa} below which formalizes the intuition that, under the given premises, the cell/block pairs $\tuplel{q,c}\in\ExX^\downarrow$ available when extracting $\Sa^{\#}$ in line~\ref{line:ComputeAbs:Se} of \REFalg{alg:GSSA} actually induce a sound abstraction relation between $\SaK$ and $\Sa^{\#}$. I.e., we always have $\SaK\frr{\alphat}{}\Sa^{\#}$ for
\begin{equation}\label{equ:alphat}
 \alphat(c):=\SetComp{q\in\Xa^{\#}}{\tuplel{q,c}\in\ExX^\downarrow}.
\end{equation}
Further, \REFprop{prop:SaKbsimSa} shows that $\SaK\frrE{\alphat}{}\Sa^{\#}$ if $\SaK$ is finite-state (and thereby $\Cover$ has stabilized from \REFlem{lem:TerminatingGSSAifSSA} in the appendix).
With this result \REFthm{thm:soundrel:Ssharp} becomes a simple corollary of \REFprop{prop:SaKbsimSa} and \REFprop{prop:SSA_sound} by utilizing the compositionality of sound abstractions %, i.e., if $(S_1,\Omapn{}{1})\frr{\alpha_{12}}{}(S_2,\Omapn{}{2})$ and $(S_2,\Omapn{}{2})\frr{\alpha_{23}}{}(S_3,\Omapn{}{3})$ then $(S_1,\Omapn{}{1})\frr{\alpha_{13}}{}(S_3,\Omapn{}{3})$ with $\alpha_{13}=\alpha_{23}\circ\alpha_{12}$
(see \REFprop{prop:compo} in the appendix for a formal proof).

\begin{proposition}\label{prop:SaKbsimSa}
 Given the premises of \REFthm{thm:soundrel:Ssharp}, it holds that
 $\SaK\frr{\alphat}{}\Sa^{\#}$ with $\alphat$ as in \eqref{equ:alphat}.  Further, if $\Cover$ has stabilized, we additionally have $\SaK\frrE{\alphat}{}\Sa^{\#}$.
\end{proposition}

\begin{proof}
To simplify notation we use $\St:=\SaK$ and $\Sa:=\Sa^{\#}$.\\ %We show the claim in seperate steps.\\
\begin{inparaitem}[$\blacktriangleright$]
 \item We first show that equality holds for (A1) and (A3) from \REFdef{def:SoundAbs}.
 \begin{inparaitem}[$\triangleright$]
  \item (A1): Observe that line~\ref{alg:SSA:Xao} in \REFalg{alg:SSA} and line~\ref{alg:GSSA:Xao} in \REFalg{alg:GSSA} literally match. Further, for all $\xa\in\Xoa$ we have that $\tuplel{\varepsilon,\xa,\xa}$ is in the initial cover set (line \ref{alg:GSSA:CoverInit} in \REFalg{alg:GSSA}) and thereby $\tuplel{\xa,\xa}\in\ExX^\downarrow$, as we have assumed $\Xo$ to respect $H$. As \REFalg{alg:GSSA} always maintains $\xt\subseteq \xa$ for any $\tuplel{\xa,\xt}\in\ExX$ and all elements in $\Cover$ only get refined, we see that there is no other $\xa'\in \Xa$ related to $\xt\in \Xt$. We therefore have $\alphat(\Xto)=\Xoa$.
   \item (A3): It is easy to see that for all $\xa\in\Xa$ holds that $x,x'\in\xa$ implies $H(x)=H(x')=\Ha(\xa)$. As $\xt\subseteq\xa$ for all $\tuplel{\xa,\xt}\in\ExX$, we have $\Ht(\xt)=\Ha(\xa)$ for all related states. %As the observation maps coincide, this shows that $\Omap{H(\xa)}=\set{\Omap{\Ha(\xa)}}$.
 \end{inparaitem}\\
 \item Now we show that (A2) holds with equality for all $\tuplel{\xt,\xa}\in\ExG$ (possibly a subset of $\ExX^\downarrow$). For this, observe that $\Sa$ is extracted in the last iteration of the while loop in line~\ref{alg:GSSA:startwhile}-\ref{alg:GSSA:endwhile} of \REFalg{alg:GSSA} and therefore the recursive function \textsc{Refine} was applied to all $\tuplel{\xa,\xt}\in\ExG$ with $\xt\subset \xa$ and has terminated. We can therefore utilize \REFlem{lem:singlecover} in the appendix implying $|\alphat(\xt)|=1$ for all $\xt$ present in $\ExG$.\\
 \begin{inparaitem}[$\triangleright$]
  \item (A2) for $\ExG$: Pick $\xt\in\Xt$, $u\in U$ and $\xt'_y=F(\xt,u)\cap H^{-1}(y)$. Further, define $Y'=\SetComp{y\in Y}{\xt'_y\neq\emptyset}$ and let $Q'$ contain all $\xa'\in\Xa$ s.t.\ $\tuplel{\xa',\xt'_y}\in\ExX$ and $y\in Y'$. Using the same argument as in the proof of \REFprop{prop:SSA_sound} we have $\Ft(\xt,u)=\bigcup
   _{y\in Y'}{\set{\xt'_y}}$, and therefore, by definition, $\alphat(\Ft(\xt,u))=Q'$.     
    Now one can verify, by looking at line~\ref{line:GSSA:cp}, \ref{line:GSSA:ExF} and \ref{line:GSSA:PostQ} of \REFalg{alg:GSSA}, that $Q'=\mathtt{PostQ}_u(\tuplel{\xa,\xt})$ for $\set{\xa}=\alphat(\xt)$. 
   Further, we extract $\Sa$ after all covers have been refined. With this we know that $F(\xa,u)=\mathtt{PostQ}_u(\tuplel{\xa,\xt})$, as otherwise there would exists a refinement $s\subset \xa$ in the sense of line~\ref{line:s} in \REFalg{alg:GSSA}. This further implies that for all $\tuplel{\xa,\xt_1},\tuplel{\xa,\xt_2}\in\ExG$ we have that  $\mathtt{PostQ}_u(\tuplel{\xa,\xt_1})=\mathtt{PostQ}_u(\tuplel{\xa,\xt_2})$. With this it follows that $Q'=\Fa(\alpha(\xt),u)$. This implies  $\alphat(\Ft(\xt,u))=\Fa(\alphat(\xt),u)$.\\
 \end{inparaitem}
 \item It remains to show that (A2) holds (with equality for a stable cover and with inclusion for an unstable one) for tuples $\tuplel{\xa,\xt}\in\ExX^\downarrow\setminus\ExG$. First, one can verify that $\tuplel{\xa,\xt}\in\ExX^\downarrow\setminus\ExG$ if (a) a tuple $\tuplel{\sigma,\xa,\xt}$ is added to $\ExX$ in the last iteration of the while loop before extracting $\Sa^{\#}$, and (b) if there exists no tuple $\tuplel{\xa',\xt}\in\ExG$ for an arbitrary $\xa'$. 
 While (a) is obvious, we show that (b) also holds. It follows from \REFlem{lem:singlecover}, that after completing every iteration of the while-loop in line~\ref{alg:GSSA:endwhile} it holds for every $\xt$ already constructed, that there exists a unique $\xa'$ s.t.\ $\tuplel{\xa',\xt}\in\ExG$. Now assume that $\tuplel{\sigma,\xa,\xt}$ is added to $\ExX$ via line~\ref{line:GSSA:Qp} of \REFalg{alg:GSSA}. Then we know that $\xa'=\xa$, as $\xa'$ is the unique minimal element of $\Cover$ covering $\xt$ and, hence, $\tuplel{\xa,\xt}\notin\ExX^\downarrow\setminus\ExG$.\\
\begin{inparaitem}[$\triangleright$]
  \item (A2) for $\ExX^\downarrow\setminus\ExG$ with stabilized $\Cover$: If $\Cover$ has stabilized no element in $\Cover$ will be further refined by \textsc{Refine}. In particular, this implies that $\xa$ is stable for any $\tuplel{\xa,\xt}\in\ExX^\downarrow\setminus\ExG$. Further, a stable cover implies that there already exists another tuple $\tuplel{\xa,\xt'}\in\ExG$ for which all %inputs are enabled and all 
  outgoing transitions are contained in $\ExF$. With this, we use the same reasoning as for $\ExG$ to construct $Q'$ and to show that (A2) holds with equality. \\%Further, it trivially follows that all inputs are enabled in $\tuplel{\xa,\xt}$.
\item (A2) for $\ExX^\downarrow\setminus\ExG$ with unstable $\Cover$:
If the $\Cover$ is not stable, we cannot ensure that $\xa$ is stable for any $\tuplel{\xa,\xt}\in\ExX^\downarrow\setminus\ExG$, i.e., would not be refined in the next iteration of the while loop. Further, we have to make sure that there exists another tuple $\tuplel{\xa,\xt'}\in\ExG$. Now recall that we initialize $\Cover$ with the largest subsets $\xa^y_0\subseteq X$ that generate the same output $y$. As $Y^{\#}=\Reach{Y}$, we know that all initial cover cells $\xa^y_0$ with $y\in\Reach{Y}$ will be explored (and possibly refined) at least once in \REFalg{alg:GSSA}. As $\xa\in \Cover$ and by construction $\xa\subseteq \xa^y_0$ for $y=H(\xa)\in\Reach{Y}$ we know that $\tuplel{\xa,\xt'}\in\ExG$. %With this it follows from the previous arguments that all inputs are enabled at $\tuplel{\xa,\xt'}$. 
With this we can use the same reasoning as in the proof of (A2) for $\ExG$ to construct $Q'$. If it is stable, the argument reduces to the previous one. If it is not, we have $F(\xa,u)\subset\mathtt{PostQ}_u(\tuplel{\xa,\xt})$. With this, the same arguments as in 
the proof of (A2) for $\ExG$ show that (A2) holds with inclusion, i.e., $\alphat(\Ft(\xt,u))\subseteq\Fa(\alphat(\xt),u)$ where $\alphat(\xt)$ contains all minimal $\xa$'s covering $\xt$. 
% Further, it follows that all inputs are enabled in $\tuplel{\xa,\xt}$ as $\tuplel{\xa,\xt'}$ is fully enabled.
\end{inparaitem}
\end{inparaitem}
\end{proof}

\begin{proof}[Proof of \REFthm{thm:soundrel:Ssharp}]
 As sound abstractions compose in the expected way (see \REFprop{prop:compo} in the appendix), we obtain a chain of sound abstractions
% \begin{equation}\label{equ:chain}
$
 \Se\frr{\LastSn{\SaK}{}}{}\SaK\frr{\alphat}{}\Sa^{\#}
$
% \end{equation}
from \REFprop{prop:SaKbsimSa} and \REFprop{prop:SSA_sound}, implying $\Se\frr{\alpha}{}\Sa$ with $\alpha=\alphat\circ\LastSn{\SaK}{}$. 
% \new{
It can be further observed from the tree-structure generated by KAM that every external prefix $\nu$ of $S$ corresponds to a unique tuple $\tuple{q,c}\in\ExX^\downarrow$. Further, the same external prefix $\nu$ reaches the state $c$ of $\SaK$ and the state $q$ of $\Sa^\sharp$. As \REFprop{prop:SaKbsimSa} shows that these states $c$ and $q$ are related via $\alphat$, we have $\LastSn{\Sa}{}=\alphat\circ\LastSn{\SaK}{}$.
% }
% As further follows from \RM{what is (A3) of Prop 4.7?} 
% (A3) of \REFprop{prop:SaKbsimSa} that $\LastSn{\Sa}{}=\alphat\circ\LastSn{\SaK}{}$, 
With this, the first claim of \REFthm{thm:soundrel:Ssharp} follows. 
The second claim follows similarly.
% by the inverse relations in \eqref{equ:chain}.
\end{proof}

\smallskip 
\noindent 
\textbf{Iterative ABCD with KAM.}
By combining \REFcor{cor:ABCDsound} and \REFthm{thm:soundrel:Ssharp} we can compute an output-feedback controller $\C:=\Ca\circ\LastSn{\Sa^\#}{}\in\WIN(S,\psi)$ from an abstract state-feedback controller $\Ca^\dagger\in\WIN^\dagger(\Sa^{\sharp},\psi)$ whenever the latter synthesis problem allows for such a solution, i.e., $\WIN^\dagger(\Sa^{\#},\psi)\neq\emptyset$. Hence, ABCD with output feedback is sound in this case. Given that $\Sa^\#$ is in general only known to abstract $S$, we are however losing completeness. 
That is, if 
$\WIN^\dagger(\Sa^{\#},\psi)=\emptyset$, it does not imply that there is no solution to the original synthesis problem $\tuplel{S,\psi}$. 

% Inspired by the framework of $l$-complete abstractions \cite{moor1999supervisory}, 
We can however take an eager abstraction-refinement approach instead 
to retain relative completeness. 
That is, whenever $\WIN^\dagger(\Sa^{\#},\psi)=\emptyset$, we run KAM for some more steps, extract a new abstraction $\Sa^{\#'}$, 
and again try to synthesize a controller. We give up, once an upper bound $L$ on the iterations of KAM is reached. 
This eager approach relies on the insight that abstractions extracted after more iterations of KAM refine earlier abstractions 
as formalized in \REFthm{thm:soundrel:SsharpL}. 
Further, this abstraction-refinement procedure is relative complete.
% \RM{the following sentence reads strange. What does "allows to solve" mean? our approach finds "it" = $L$?}
That is, if there is a topologically closed finite-state abstraction $\Sa$ for which $\WIN(\Sa, \psi)\neq\emptyset$, 
there always exists a large enough $L$ s.t.\ the abstraction $\Sa^{\#}$ extracted from KAM in the $L$'s iteration allows to solve the controller synthesis problem, i.e., $\WIN(\Sa^{\#}, \psi)\neq\emptyset$.

\begin{theorem}\label{thm:soundrel:SsharpL}
 Given the premises of \REFthm{thm:soundrel:Ssharp}, let $\Sa_{+1}^{\#}$ be the system computed in line~\ref{line:ComputeAbs:Se} of \REFalg{alg:GSSA} 
after one more iteration of \REFalg{alg:GSSA} after $\Sa^\#$ was extracted. Then
$\Sa_{+1}^{\#}\frr{}{}\Sa^{\#}$.
\end{theorem}

\begin{proof}
 Let $\ExG$, $\ExX^\downarrow$ and $\ExG'$, $\ExX^{\downarrow'}$ be the sets computed when extracting $\Sa^\#$ and $\Sa_{+1}^{\#}$, respectively. Further let us define an abstraction map candidate $\alpha_{+1}$ using three cases. I.e., $q\in\alpha_{+1}(p)$ if there exists $c$ s.t.\ either 
 (a) $\tuplel{q,c}\in\ExG$ and $q=p$, or 
 (b) $\tuplel{q,c}\in\ExX^\downarrow\setminus\ExG$, $\tuplel{p,c}\in\ExG'$ and $p\subseteq q$, or
 (c) $\tuplel{p,c}\in\ExX^{\downarrow'}\setminus\ExG'$ and there exists $c'$ s.t.\ $q$ is related to $p$ as in (a) or (b).
 
This definition induces the following three cases for the proof.
 \begin{inparaitem}[$\triangleright$]
  \item (a) holds for $(q,p)$: This implies $\tuplel{q,c}\in\ExG'$. It follows from the same arguments as used in the proof of \REFprop{prop:SaKbsimSa} that equality holds for (A1)-(A4) in \REFdef{def:SoundAbs} w.r.t.\ $\St$ both for $\Sa^\#$ and $\Sa_{+1}^{\#}$. As $\alpha_{+1}$ reduces to the identity map in this case, the claim trivially follows. \\
  \item (b) holds for $(q,p)$:  Then it follows again that equality holds for (A1)-(A4) in \REFdef{def:SoundAbs} w.r.t.\ $\St$ for $\Sa_{+1}^{\#}$ but it follows from \REFthm{thm:soundrel:Ssharp} that only inclusion holds for (A3) w.r.t.\ $\Sa^\#$. 
  Formally, we fix $c$ existentially quantified in the definition of case (b) before. Then we have  $\alphat_{+1}(\Ft(c,u))=\Fa_{+1}(\alphat_{+1}(c),u)$ where $\alphat_{+1}(c)$ contains the unique minimal $p$ covering $c$ and $\alphat(\Ft(c,u))\subseteq\Fa(\alphat(c),u)$ where $\alphat(c)$ contains all minimal $q$'s covering $c$. We have $p\subseteq q$ for all $q\in\alphat(c)$ due to the additional refinement step run before extracting $\Sa_{+1}^{\#}$. In particular, we have $\alphat(c)=\alpha_{+1}(p)$. Hence,  $\alphat_{+1}(\Ft(c,u))=\Fa_{+1}(p,u)$ and $\alphat(\Ft(c,u))\subseteq\Fa(\alpha_{+1}(p),u)$.
  Now define $C'=\Ft(c,u)$. If for all $c'\in C'$ case (a) or (b) holds, we have that $\alphat_{+1}(c')$ maps to a unique $p'$. In this case it holds that $\alpha_{+1}(\alphat_{+1}(\Ft(c,u)))=\alphat(\Ft(c,u))$ and therefore $\alpha_{+1}(\Fa_{+1}(p,u))\subseteq \Fa(\alpha_{+1}(p),u)$, what proves the statement. Now for any $c''$ for which case (c) applies there exists a $c'''$ s.t.\ case (a) or (b) applies while $\alphat(c'')=\alphat(c''')$ and $\alphat_{+1}(c'')=\alphat_{+1}(c''')$. With this, the previous argument applies and the claim follows.\\
 \item (c) holds for $(q,p)$: Fix $c$ existentially quantified in the definition of (c) and recall that there exists $c'$ s.t.\ $\alphat(c)=\alphat(c')$ and $\alphat_{+1}(c)=\alphat_{+1}(c')$ and case (a) or (b) applies for $c'$. Hence, without loss of generality we can replace $c$ by $c'$ and the claim follows.
 \end{inparaitem}
\end{proof}

\begin{remark}\label{rem:lcomplete}
 The idea of abstraction-refinement for controller synthesis is also often applied in the context of $l$-complete 
abstractions \cite{moor1999supervisory,schmuck2014asynchronous,yang2018local,reissig2011computing}. %, which are inspired by \emph{behavioral systems theory} \cite{Willems}. 
% As discussed in \REFrem{rem:conectionBehavior}, 
% These abstractions are typically constructed from the external behavior of $S$. 
Similar to KAM, $l$-complete abstractions are constructed forward and generalize from initial observations to equivalence classes. 
Here, the equivalence classes collect states which share the same $l$-long external history (see e.g., \REFfig{ex:lcomplete} for an example with $l=2$). 
%By this it generalizes from the initial observations to equivalence classes that allow to over approximate the external behavior of $\Ext{S}$ and is thereby similar to our algorithm in \REFalg{alg:GSSA}. 
% 
$l$-complete abstractions are typically constructed from the external behavior of $S$ and do not assume the state dynamics of $S$ to be known. They thereby do not utilize 
the memory structure implicitly given by the state dynamics of $S$ in their generalization step. 
Therefore, KAM generates tighter abstractions whenever the underlying state transition system is known, but $l$-complete abstractions are to be preferred if this is not the case.
 \end{remark}
 
 \begin{figure}
             \begin{tikzpicture}[auto,scale=1]
   \def\vsep{0.35}
   \def\hsep{0.15}
     \begin{footnotesize}
     \node at (-1,0) {$\Sa^l:$}; 
       \node (init) at (-0.5,0) {};  
  \node [myestate, draw] (a1) at (0,0) {A}; \node[] at ($(a1.east)+(\hsep,0)$) {$A$};
  \node [myestate, draw] (a2) at (-0.7,-0.7) {AA}; \node[] at ($(a2)+(0,-\vsep)$) {$A$};
  \node [myestate, draw] (b1) at (0.3,-0.7) {AB}; \node[] at ($(b1)+(+\hsep,+\vsep)$) {$B$};
\node [myestate, draw] (b3) at (1.2,-0.7) {BA}; \node[] at ($(b3)+(0,-\vsep)$) {$A$};
  \node [myestate, draw] (b2) at (2.1,-0.7) {BB}; \node[] at ($(b2)+(0,-\vsep)$) {$B$};  
%   \node [] (dum) at (0.7,-2.1) {};
%   \node [myestate, draw] (b2) at (0.7,-2.8) {$\set{s_{n-1},s_{n+1}}$}; \node[] at ($(b2)+(0,\vsep)$) {$b$};

\SFSAutomatEdge{init}{}{a1}{}{}
\SFSAutomatEdge{a1}{}{a2}{}{}
\SFSAutomatEdge{a2}{}{a2}{loop left}{}
\SFSAutomatEdge{a1}{}{b1}{}{}
\SFSAutomatEdge{b1}{}{b3}{}{}
\SFSAutomatEdge{b1}{}{b2}{bend left}{}
\SFSAutomatEdge{b3}{}{a2}{bend left}{}
\SFSAutomatEdge{b2}{}{b3}{}{}
\SFSAutomatEdge{b2}{}{b2}{loop right}{}
             \end{footnotesize}
         \end{tikzpicture}
         \caption{$2$-complete abstraction of the system $S$ in \REFfig{fig:ex_inf_GK}.}\label{ex:lcomplete}
 \end{figure}
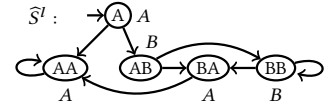

\noindent
\textbf{Symbolic Implementations.}
%
% We note that our algorithm can be implemented in a symbolic manner on top of a region algebra \cite{HenzingerMR05}
% that allows computing post operations and intersections.
KAM differs from the simultaneous reachability and bisimulation minimization algorithm of Lee and Yannakakis \cite{LeeYannakakis92} 
as it constructs an external language- (not bisimulation-) equivalent system.
Hence, it only applies predecessor operations and intersection with outputs, but does not take set differences. 
This is in fact crucial in implementations.
For example, for affine systems with polyhedral initial sets and output sets, one can implement the algorithm exactly
using a convex polyhedral abstract domain, as both predecessor operators and intersections maintain convexity while 
set differences do not.

\section{Hybrid System Examples}\label{sec:exs}
We now present two continuous-state discrete-time hybrid system examples and show how our approach
can be used to design abstractions useful for output-feedback control. Along the way, 
we also compare our approach with several alternatives and show how 
state-of-the-art techniques for abstracting continuous-state systems, 
such as those implemented in SCOTS or Mascot \cite{SCOTS,HsuMMS18}, can be incorporated in our approach.

\begin{example}\label{ex:Sigma_1}
Consider a switched system $\Sigma_1$ with 
\begin{inparaitem}[$\triangleright$]
\item state space $X = [0,3)\times[0,3) \subset \real{2}$; 
\item initial states $X_0 = X$;
\item input space $U=\Set{u_1, u_2}$ (corresponding to two controllable modes);
\item  output space $Y = \Set{y_{00},y_{01},y_{02},y_{10},y_{11},y_{12},y_{20},y_{21}, y_{22}}$; %$Y = \Set{A,B,C,D,E,F,G, H, I}$ 
\item output function $H: x \mapsto y_{ij}$, where $i=\floor{x_1/3}$, $j=\floor{x_2/3}$ for all $x\in X$  ; and
\item transition function $F$ defined as
$$F(x,u_1) = {mod}_3 \left( x + \begin{bmatrix} 0.4 \\ 0.4\end{bmatrix}\right), \; F(x,u_2) = {mod}_3 \left( x - \begin{bmatrix} 0.4 \\ 0.4\end{bmatrix}\right),$$
where the function ${mod}_k: \real{n} \rightarrow [0,k)^n$ wraps its input argument component-wise around the perimeter of its codomain; i.e., 
 if $s = {mod}_k(x)$, then $s_i = x_i - k\floor{\frac{x_i}{k}}$. 
%\begin{equation}
%x' = x + K_u \;\text{ with }\; K_{u_1} = \begin{bmatrix} 0.3 \\ 0.3\end{bmatrix}, \quad K_{u_2} = \begin{bmatrix} 0.3 \\ 0.3\end{bmatrix}
%\end{equation}
\end{inparaitem}
In \REFfig{fig:sigma1} (top left), state space $X$ is shown, where the domain of $H$ for all $y$ is indicated by the large boxes with edge length $1$. The dynamics of $F$ are then interpreted as upward ($u=u_1$) and downward ($u=u_2$) discrete-time flows of points in $X$ parallel to the diagonal connecting the lower left and top right corner of $X$. When the boundary of $X$ is reached, the system continues to evolve in the block reached by wrapping $X$ around its boundaries. 
Note that the only source of non-determinism in system $\Sigma_1$ is due to the initial condition not being a singleton, whereas the transition function is deterministic.
% 
% We interpret the trajectories of the system over atomic propositions $\AP = \Set{{goal}_1, {goal}_2}$, with the observation
% map 
% $\Omap{}_1(y_{ij}) = \Set{{goal}_1}$ if $ij = 00$, 
% $\Omap{}_1(y_{ij}) = \Set{{goal}_2}$ if $ij = 22$, and
% $\Omap{}_1(y_{ij}) = \emptyset$ otherwise.
% 
% $$\Omap{}_1(y_{ij}) = \begin{cases} 
% \Set{{goal}_1} & \text{ if }\;  $ij = 00$, \\  
% \Set{{goal}_2} & \text{ if }\; $ij = 22$, \\
% \emptyset & \text{ otherwise}.
% \end{cases} $$
We consider a specification $\psi_1$ stating that when starting in $y_{00}$ the system should always eventually (re-)visit $y_{00}$ and $y_{22}$. %leading to 
% $\psi_1 = y_{00}\rightarrow (\aeventually{y_{00}} \conj \aeventually{y_{22}})$
% the control problem $\tuplel{\Sigma_1,\psi_1}$.
\end{example}

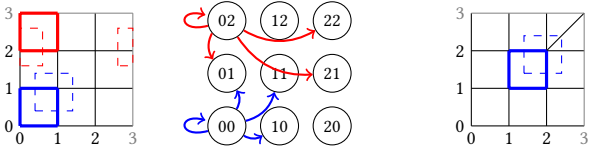
\begin{figure}[t]
   \begin{center}
%!TEX root = main.tex
\newcommand*{\xMin}{0}%
\newcommand*{\xMax}{3}%
\newcommand*{\xMaxB}{2}%
\newcommand*{\yMin}{0}%
\newcommand*{\yMax}{3}%
\newcommand*{\yMaxB}{2}%

\begin{tikzpicture}[auto,scale=0.5]
\begin{footnotesize}
    \foreach \i in {\xMin,...,\xMaxB} {
        \draw [] (\i,\yMin) -- (\i,\yMax)  node [below] at (\i,\yMin) {$\i$};
    }
    \foreach \i in {\yMin,...,\yMaxB} {
        \draw [] (\xMin,\i) -- (\xMax,\i) node [left] at (\xMin,\i) {$\i$};
    }
\draw[very thin, gray] (\xMax,\yMin) -- (\xMax,\yMax)  node [below] at (\xMax,\yMin) {$\xMax$};
\draw[very thin, gray] (\xMin,\yMax) -- (\xMax,\yMax) node [left] at (\xMin,\yMax) {$\yMax$};
\draw [step=1.0,blue, very thick] (0,0) grid (1,1);
\draw [dashed, blue] (0.4,0.4) -- (1.4,0.4) -- (1.4,1.4) -- (0.4,1.4) -- (0.4,0.4);
\draw [step=1.0,red, very thick] (0,2) grid (1,3);
\draw [dashed, red] (0,1.6) -- (0, 2.6) -- (0.6,2.6) -- (0.6,1.6) -- (0,1.6);
\draw [dashed, red] (2.6,1.6) -- (2.6, 2.6) -- (3,2.6) -- (3,1.6) -- (2.6,1.6);
%\draw [very thick, brown, step=1.0cm,xshift=-0.5cm, yshift=-0.5cm] (0.5,0.5) grid +(2.5,1.5);
% \end{footnotesize}
% \end{tikzpicture}
%%%%%%%%%%%%%%%%%%%%%%%%%%%%%%
% \begin{tikzpicture}[node distance=0.8cm, scale=0.8]
% \begin{footnotesize}
\def\hsep{5.5}
%  	\tikzstyle{every node}=[circle,draw]
	\foreach \i in {0, ..., 2} {
		\foreach \j in {0, ..., 2} {
			
			\node[circle,draw] (a\i\j) at (\hsep+1.4*\i, 1.4*\j) {\i\j};
		}

	}
	\draw[thick, ->, blue]
									     (a00) edge[loop left] (a00)
									     (a00) edge[bend right] (a11)
									     (a00) edge[bend right] (a10)
									     (a00) edge[bend right] (a01);

	\draw[thick, ->, red] 
									      (a02) edge[loop left] (a02)
									     (a02) edge[bend right] (a01.north west)
									     (a02) edge[bend right] (a22)
									     (a02) edge[bend right] (a21);
% \end{footnotesize}
% \end{tikzpicture}
%%%%%%%%%%%%%%%%%%%%%%%%%%%%%%
% \begin{tikzpicture}[auto,scale=0.6]
%    \def\sc{0.7}
%     \foreach \i in {\xMin,...,\xMaxB} {
%         \draw [] (\i,\yMin) -- (\i,\yMax)  node [below] at (\i,\yMin) {$\i$};
%     }
%     \foreach \i in {\yMin,...,\yMaxB} {
%         \draw [] (\xMin,\i) -- (\xMax,\i) node [left] at (\xMin,\i) {$\i$};
%     }
% \draw[very thin, gray] (\xMax,\yMin) -- (\xMax,\yMax)  node [below] at (\xMax,\yMin) {$\xMax$};
% \draw[very thin, gray] (\xMin,\yMax) -- (\xMax,\yMax) node [left] at (\xMin,\yMax) {$\yMax$};
% \draw [very thin, step=0.2, gray] (0,0) grid +(3,3); 
% \draw [step=0.2,blue, thick] (0,0) grid (0.2,0.2);
% \draw [step=0.2,blue, densely dashed] (0.4,0.4) grid (0.6,0.6);
% \draw [red, thick] (0,2) -- (0,2.2) -- (0.2,2.2) -- (0.2,2) -- (0,2);
% \draw [densely dashed, red] (2.6,1.6) -- (2.6, 1.8) -- (2.8,1.8) -- (2.8,1.6) -- (2.6,1.6);
% %\draw [step=0.2,red, dashed] (0.4,0.4) grid (0.6,0.6);
% %\draw [dashed, blue] (0.4,0.4) -- (1.4,0.4) -- (1.4,1.4) -- (0.4,1.4) -- (0.4,0.4);
% %\draw [step=1.0,blue, very thick] (0.5,0.5) grid (2.5,1.5);
% %\draw [very thick, brown, step=1.0cm,xshift=-0.5cm, yshift=-0.5cm] (0.5,0.5) grid +(2.5,1.5);
% \end{tikzpicture}
%%%%%%%%%%%%%%%%%%%%%%%%%%%%%%%%%
% \hspace{-0.2cm}
% \begin{tikzpicture}[auto,scale=0.5]
% \begin{footnotesize}
\def\hsepn{12}
    \foreach \i in {\xMin,...,\xMaxB} {
        \draw [] (\hsepn+\i,\yMin) -- (\hsepn+\i,\yMax)  node [below] at (\hsepn+\i,\yMin) {$\i$};
    }
    \foreach \i in {\yMin,...,\yMaxB} {
        \draw [] (\hsepn+\xMin,\i) -- (\hsepn+\xMax,\i) node [left] at (\hsepn+\xMin,\i) {$\i$};
    }
\draw[](\hsepn+2,2)--(\hsepn+3,3);
\draw[very thin, gray] (\hsepn+\xMax,\yMin) -- (\hsepn+\xMax,\yMax)  node [below] at (\hsepn+\xMax,\yMin) {$\xMax$};
\draw[very thin, gray] (\hsepn+\xMin,\yMax) -- (\hsepn+\xMax,\yMax) node [left] at (\hsepn+\xMin,\yMax) {$\yMax$};
\draw [step=1.0,blue, very thick] (\hsepn+1,1) grid (\hsepn+2,2);
\draw [step=1.0,blue, dashed] (\hsepn+1.4,1.4) rectangle (\hsepn+2.4,2.4);

% \draw [red, very thick] (\hsepn+2,2) -- (\hsepn+3, 3) -- (\hsepn+2,3) -- (\hsepn+2,2);
% \draw [red, dashed] (\hsepn+1.6,1.6) -- (\hsepn+2.6, 2.6) -- (\hsepn+1.6,2.6) -- (\hsepn+1.6,1.6);
%\draw [blue, very thick] (2,2) -- (3, 2) -- (3,3) -- (2,2);
%\draw [blue, dashed] (2.4,2.4) -- (3, 2.4) -- (3,3) -- (2.4,2.4);
\end{footnotesize}
\end{tikzpicture}

%%%%%%%%%%%%%%%%%%%%%%%%%%%%%%%%%
\end{center}
\vspace{-0.2cm}
\caption{Graphical representation of $\Sigma_1$ (far left) and $\Sigma_2$ (far right), showing the state space $X$ with the partition induced by the output maps $H_1$ and $H_2$, respectively. For $\Sigma_1$, $X_j=F(X_i,u_1)$ (dashed blue) indicates the reachable set of $X_i=H_1^{-1}(y_{00})$ (solid blue). Intersecting $X_j$ with the partition generates transitions (blue) originating in $y_{00}$ in the finite-state abstraction (middle). Similarly, $X_j=F(X_i,u_2)$ (dashed red) is reached from $X_i=H_1^{-1}(y_{02})$ (solid red) generating transitions (red) originating in $y_{02}$ in the abstraction. 
% This abstraction process generates non-determinism that make the posed abstract synthesis problem unsolvable. Algorithm KAM automatically induces a grid of size $\eta=0.2$ that results in a sound finite-state realization of $\Sigma_1$. 
}\label{fig:sigma1}
\vspace*{-0.6cm}
\end{figure}

Let us first consider constructing an abstract system $\widehat{\Sigma}_1$ 
that has a feedback refinement relation (FRR) with $\Sigma_1$ by using forward simulation as, e.g., implemented in SCOTS. 
The main idea is to ``grid'' the state space into hyperboxes of size $\eta$ in a way consistent with the 
outputs
and treat each grid cell $X_i$ as an abstract state. 
Then for each grid cell $X_i$ and for each input $u_i$, post $F(X_i,u_i)$ is computed and a transition with input label $u_i$ is 
added from the abstract state $X_i$ to all abstract states $X_j$ that have a non-empty intersection with the post. 
This process is illustrated in \REFfig{fig:sigma1}. 
% for a grid size $\eta=1$ where the dashed blue box represents $X_j=F(X_i,u_1)$ with $X_i=H^{-1}(y_{00})$ (solid blue). By intersecting $X_j$ with the hyperboxes, we get the blue transitions in the abstract, finite state abstraction on the left. Similarly, the red boxes and transitions are induced by the transition $X_j=F(X_i,u_2)$ with $X_i=H^{-1}(y_{02})$.
% 
Given the existence of an FRR from $\Sigma_1$ to $\widehat{\Sigma}_1$ (rendering $\widehat{\Sigma}_1$ a sound abstraction of $\Sigma_1$ for state-feedback control as discussed in \REFrem{rem:FRR}) and the compositionality of sound abstractions (see \REFprop{prop:compo}), we can use $\widehat{\Sigma}_1$
with any of the algorithms presented in \REFsec{sec:algos} to construct an abstraction $\widehat{\Sigma}'_1$ 
which allows to solve the \emph{output-feedback control problem} over $\Sigma_1$.

In order to apply this process, we need to select a grid size $\eta$ when constructing $\widehat{\Sigma}_1$. 
We denote the resulting abstraction with $\widehat{\Sigma}_1^{(\eta)}$. 
We can start with $\eta = 1$ as discussed before. 
This, however induces non-determinism and it can be easily seen by inspecting \REFfig{fig:sigma1} (middle), 
that there does not exist a controller in the abstraction that allows us to surely transition from $y_{00}$ to $y_{22}$ and back infinitely often---in the abstraction,
applying the necessary input sequence might lead to  visiting $y_{02}$ instead of $y_{22}$.
One can try a finer grid size, e.g., $\eta = 0.03$, but the problem still does not admit a solution for $\tuplel{\widehat{\Sigma}_1^{(0.03)},\psi_1}$. 
By inspection, the problem only has a solution if $\eta$ is chosen such that $0.2$ is an integer multiple of $\eta$. 
Here, $0.2$ is the greatest common divisor of $0.4$ (the increments the dynamics make) and $1$ (the ``fidelity" of the outputs). 
So, the set of grid sizes that gives a solution is a measure-zero set in $\real{}_{>0}$ and, in general, the ``right" grid size is dictated by the 
dynamics and output map. 
Further, even if we use an automatic refinement tool like Mascot, the step size of the refinement of $\eta$ is a design parameter 
and thus, the tool may not ever explore an integer multiple of $0.2$. 
% \AKS{Necmiye, what would your refinement technique do here? Would if work?} 

We now turn to solving the output-feedback control problem $\tuplel{\Sigma_1,\psi_1}$ by directly applying the algorithms discussed 
in \REFsec{sec:algos} to $\Sigma_1$ without constructing $\widehat{\Sigma}_1$ first. 
For this example, all three algorithms (i.e., KA, KA with bisimulation quotient, and KAM) will produce the same abstraction. 
This is due to the fact that the dynamics of the system are such that the post and the pre operations over $F$ cancel out. 
Therefore the forward and backward algorithms are essentially performing the same operations. 
Further, all of them terminate and generate a sound \emph{realization}. 
Thus, these algorithms \emph{automatically} figure out that the largest cover of $X$ which merges states with the same 
future under any applied input sequence has size $\eta=0.2$.

% We will now slightly modify \REFexp{ex:Sigma_1}.
%  \NO{Check: if we need to specify $\Xo$ differently (or update the spec to include initial similar to $\psi_1$).}
% \RM{sure, but I think it is ok to gloss over the difference rather than make this part def. heavy}
\begin{example}
We consider another switched system $\Sigma_2$ with the same dynamics as $\Sigma_1$ but 
with changed output space $Y_2=\Set{y_{00},\hdots, y_{21}, y_{22u},y_{22l}}$ s.t.\ $H_2$ maps the 
upper left and lower right triangle of $y_{22}$ to $y_{22u}$ and $y_{22l}$, respectively (see \REFfig{fig:sigma1} (right) for an illustration). 
% \begin{itemize}
% \item state space $X_2 = \real{2}$; 
% \item initial states $X_0 = [0,1)\times[0,1)$;
% \item input space $U=\Set{u_1, u_2}$ corresponding to two controllable modes;
% \item  output space $Y = \Set{00,01,02,10,11,12,20,21, 22, 33}$; %$Y = \Set{A,B,C,D,E,F,G, H, I}$ 
% \item output function $H: x \mapsto ij$, where $i=\floor{x_1/3}$, $j=\floor{x_2/3}$ if $x\in [0,3)\times[0,3)$, and $H: x \mapsto 33$, otherwise;
% \item transition function $F$ defined as
% $$F(x,u_1) =  x + \begin{bmatrix} 0.4 \\ 0.4\end{bmatrix}, \; F(x,u_2) =  x - \begin{bmatrix} 0.4 \\ 0.4\end{bmatrix}.$$
% \end{itemize}
% \NO{I can index the components of the switched system $\Sigma_2$ to differential them from the first ones but I wanted to avoid the clutter. Let me know if you prefer indexing.}
The specification $\psi_2$ requires to repeatedly visit $y_{00}$ and either $y_{22u}$ or $y_{22l}$ infinitely often after starting in $y_{00}$.
% Similarly, we consider two different propositions ${goal}_{2,u}$ and ${goal}_{2,l}$ which are mapped to by $\Omap{}_2$ from $y_{22u}$ and $y_{22l}$, respectively.
% Consider an LTL specification $\psi_2 = {goal}_1\rightarrow (\aeventually{{goal}_1} \conj \aeventually{(\propDisj{{goal}_{2l}}{{goal}_{2u}})})$. 
% Thus, the control problem is $\tuplel{\Sigma_2,\Omap{}_2,\psi_2}$.
\end{example}

Consider running KAM on $\Sigma_2$. First observe that we are now initializing KAM with the triangle shape domains of $H(y_{22l})$ and $H(y_{22u})$ 
in addition to the the boxed domains for all remaining outputs. 
This will result in little triangles right above and right below the diagonal of $y_{33}$, 
which collect reachable state subsets with the same output. However, in the remaining part of the state space, KAM will converge to the same rectangular grid as it does for $\Sigma_1$. The intuitive reason for this is that the post of any set $H^{-1}(y)$ with $y\notin\set{y_{22l},y_{22u}}$ remains a box. 
Therefore, we can never distinguish whether we observe $y_{22u}$ or $y_{22l}$ if we transition to a box on the diagonal of $y_{33}$, no 
matter how fine we grid. 
Further, the post of any such box will be either $\set{y_{22u},y_{22l}}$ again, $\set{y_{00}}$ (for $u=u_1)$ or $\set{y_{22}}$ (for $u=u_2)$. 
With this it is easy to see that boxes of size $\eta=0.2$ are again the largest partition of $X$ that form equivalence classes respecting observable subsets. KAM will therefore compute the same sound realization for $\Sigma_2$ as for $\Sigma_1$. If we however run KA (with or without the bisimulation quotient) one would additionally chop every box of size $\eta=0.2$ into an upper left and lower right triangle. This unnecessary doubles the state space of the abstraction, but still resulting in a sound realization.

Let us now consider computing an abstraction $\widehat{\Sigma}_2^{(\eta)}$ by forward simulation of $\Sigma_2$ first, using SCOTS. 
Then we immediately get into trouble, because we cannot find a rectangular grid that respects the output map, 
as needed to fulfill (A3) in \REFdef{def:SoundAbs}. This approach would therefore directly fail in this example.

Finally, consider a system $\Sigma_3$ which has an unbounded state space $X_3=\real{2}$ with 
transition function defined by $F$ of $\Sigma_1$ but without the wrapping of its input argument. 
The output set $Y_3$ and the output function $H_3$ of $\Sigma_3$ are given by tiling the entire $\real{2}$ space irregularly with the 3x3 blocks of observations $Y_1$ and $Y_2$ 
along with their respective output maps $H_1$ and $H_2$. 
We still have a finite set of inputs and outputs. By recalling that KAM produces the same sound realization for $\Sigma_1$ and $\Sigma_2$, we can use the 
same arguments as in the example of \REFsec{sec:KAMexp} to see that KAM will generate the same sound realization for   $\Sigma_3$ as for $\Sigma_1$ and $\Sigma_2$, while all other 
algorithms will produce infinite-state abstractions.
Admittedly, while the example distinguishes KAM from the other algorithms, it is not clear how to symbolically represent the algorithm in this case.
% While this is a nice illustrative example, it is not clear how to implement an algorithm for KAM which produces this result as the 
% representation of equivalence classes becomes very complex in this case.

% 
% \smallskip \noindent \textbf{Acknowledgements.}
\begin{acks} %% acks environment is optional, suppressed with 'anonymous'
  %% Commands \grantsponsor{<sponsorID>}{<name>}{<url>} and
  %% \grantnum[<url>]{<sponsorID>}{<number>} should be used to
  %% acknowledge financial support and will be used by metadata
  %% extraction tools.
This research was funded in part by the \grantsponsor{dfg}{DFG}{https://www.dfg.de/} project \grantnum{dfg}{389792660-TRR 248}
and by the \grantsponsor{erc}{ERC}{https://erc.europa.eu/} under the 
Grant Agreement 
\grantnum{erc}{610150}. 
% \grantnum[http://www.impact-erc.eu/]{erc}{610150}.
Ozay was supported in part by ONR grant N00014-18-1-2501, NSF grant ECCS-1553873, and an Early Career Faculty grant from NASA's Space Technology Research Grants Program.
\end{acks}

% \NO{Minor comments (can be ignored): \begin{itemize} 
% \item check redundancies in the notation section. not sure if positive reals is relevant
% \item Example numbering in sec 4.4 is odd. If possible better to  have, "Example 4.3 (cont)" as all the examples to indicate that we are continuing with the same example.
% \item Add some text in the appendix between results to ease following the logic
% \item In the example, should we say we omitted $U$ in the first element as it is singleton not to clutter the notation?
% \item is  it possible  to spare some space for the l-complete example in 4.2?
% \end{itemize}}

\balance
\bibliographystyle{abbrv}
{%\tiny
\bibliography{reportbib}}

\begin{appendix}
 \section{Additional Proofs}
% We provide the proofs for the more general case of not fully enabled inputs.

\begin{proof}[Proof of \REFthm{thm:ABCDsound}]
 We provide theis proof for the more general case of not fully enabled inputs.
 
For the first claim we pick $\pi=x_0u_0x_1u_1\hdots\in\CPaths{S,\C^\dagger}$ with external sequence $\rho=y_0u_0y_1\hdots\in\Ext{S,\C^\dagger}$ s.t. $y_k=H(x_k)$  for all $k\in\N$ and show $\rho\in\semantics{\psi}$.
% \NO{is there a reason for switching from $k$ to $i$?}
 
 For $k=0$, the definition of $\CPaths{S,\C^\dagger}$ implies that $x_0\in\Xo$. %Further, it follows from the definition of $\EPrefs{}$ that $y_0\in\EPrefs{S}=\Xe $ with $y_0\in\Xeo$.
 Using (A1) we know that for all $\xa_0\in\alpha(x_0)$ holds that  $\xa_0\in\Xoa$. We further have $H(x_0)=y_0$.
%  $\Ha(\xa_0)=\Last{\ya_0}=\ya_0$ and  $\He(y_0)=\Last{y_0}=y_0$ and $\Omap{y_0}=y_0$.
 Now it follows from (A3) that for all $\xa_0\in\alpha(x_0)$ we have $y_0\in\Ha(\xa_0))$ and therefore $y_0\in\Ext{\Sa,\Ca^\dagger}|_{[0;0]}$.
 
 For $k>0$ assume $\rho|_{[0;k-1]}\in\Ext{\Sa,\Ca^\dagger}|_{[0;k-1]}$ and show $\rho|_{[0;k]}\in\Ext{\Sa,\Ca^\dagger}|_{[0;k]}$. Let $\pi=x_0u_0x_1u_1\hdots x_{k-1}\in\CPrefs{S,C}$. Now pick any $\pia=\xa_0u_0\xa_1u_1\hdots\xa_{k-1}\in\alpha(\pi)$ %and observe that $u_{k-1}\in\Ca(\pia)=\C(\nu)$ by definition. 
%  which is in $\Enab_{\Sa^*}(\nua)$ by definition and it follows from (A2) that $u_k\in\Enab_{\Se }(\nu)$. Further, observe that $u_k\in\Ca(\alpha(\nu))=\C(\nu)$ by definition. 
%  
 and let $x_k=F(x_{k-1},u_{k-1})$ and $\xa_k=\Fa(x_{k-1},u_{k-1})$ and observe that $\pi u_{k-1} x_k\in\CPaths{S,\C^\dagger}$ and $\pia u_{k-1} \xa_k\in\CPaths{\Sa,\Ca^\dagger}$. Further, it follows from (A2) that 
%  $\pia u_{k-1} \xa_k\in\alpha(\pi u_{k-1} x_k)$ and therefore $\pi u_{k-1} x_k\in\gamma(\pia u_{k-1} \xa_k)$. As $\Omap{H(\pi u_{k-1} x_k)}=H(x_k)=y_k$ it follows from (A3) that $y_k\in\Omapa{\Ha(\pia u_{k-1} \xa_k)}$ and therefore $\rho|_{[0;k]}\in\Ext{\Sa,\Ca^\dagger}|_{[0;k]}$.
$\alpha(x_k)\subseteq \xa_k$ and therefore $x_k\in\gamma(\xa_k)$. As $H(x_k)=y_k$, (A3) implies $y_k\in\Ha(\xa_k)$ and, hence, $\rho|_{[0;k]}\in\Ext{\Sa,\Ca^\dagger}|_{[0;k]}$.
 
 As $\Ext{\Sa}$ is topologically closed, so is $\Ext{\Sa,\Ca^\dagger}$. With this $\rho|_{[0;k]}\in\Ext{\Sa,\Ca^\dagger}|_{[0;k]}$ for all $k\in\N$ implies $\rho\in\Ext{\Sa,\Ca^\dagger}$. As  $\Ca^\dagger\in\WIN^\dagger(\Sa,\psi)$, we have $\Ext{\Sa,\Ca^\dagger}\subseteq\semantics{\psi}$ and, hence, $\rho\in\semantics{\psi}$.
 
 For the second claim, one can verify that $S\frrE{\alpha}{\gamma}\Sa$ implies $\Sa\frr{\gamma}{} S$ and by this $\C'\in\WIN^\dagger(S,\psi)$ implies $\Ca'=\C'\circ\gamma\in\WIN^\dagger(\Sa,\psi)$ from the first part of this theorem. Hence, either $\WIN^\dagger(S,\psi)=\emptyset$ or $\WIN^\dagger(\Sa,\psi)\neq\emptyset$. The \enquote{only if} part follows analogously from the inverse direction.%. This %can be equivalently written as $\propImp{\WIN(\Sa,\psi)=\emptyset}{\WIN(S,\psi)=\emptyset}$, what 
%  what proves the statement.
% 
%  The last statement follows from the fact that $\Se\frrE{\alpha}{}(\Sa^*,\Omapa{})$ implies $\Sa\frr{\gamma}{} S$ and by this $\C'\in\WIN(S,\psi)$ implies $\Ca=\Ca\circ\gamma\in\WIN(\Sa,\psi)$ from the first part of this theorem. %Then equality follows from the fact that ${\xa}=\alpha(\gamma(\xa))$ if $\Se\frrE{\alpha}{}(\Sa^*,\Omapa{})$.
 \end{proof}
 
 The next proposition shows the compositionality of sound abstraction relations.

\begin{proposition}\label{prop:compo}
 Let $(S_1,\Omapn{}{1})\frr{\alpha_{12}}{}(S_2,\Omapn{}{2})$ and $(S_2,\Omapn{}{2})\frr{\alpha_{23}}{}(S_3,\Omapn{}{3})$ then $(S_1,\Omapn{}{1})\frr{\alpha_{13}}{}(S_3,\Omapn{}{3})$ with $\alpha_{13}=\alpha_{23}\circ\alpha_{12}$. 
\end{proposition}

\begin{proof}
We show that (A1)-(A3) in \REFdef{def:SoundAbs} hold by using the observation that $\alpha_{12}(x_1)\subseteq X_2$ and $\alpha_{23}(x_2)\subseteq X_3$ for $x_1\in X_1$, $x_2\in X_2$. Further, we define $\gamma_{ji}$ as the induced inverses of the respective $\alpha_{ij}$.\\
 \begin{inparaitem}[$\blacktriangleright$]
  \item (A1) As $\alpha_{12}(X_{1,0}) \subseteq X_{2,0}$ and $\alpha_{23}(X_{2,0}) \subseteq X_{3,0}$ it follows that $\alpha_{23}(\alpha_{12}(X_{1,0})) \subseteq X_{3,0}$. \\
  \item (A2.1) As $\Enab_{S_3}(\alpha_{23}(x_2))\subseteq\Enab_{S_2}(x_2)$ and $\Enab_{S_2}(\alpha_{12}(x_1))\subseteq\Enab_{S_1}(x_1)$ for any $x_1\in X_1$ and  $x_2\in X_2$ it follows that $\Enab_{S_3}(\alpha_{23}(\alpha_{12}(x_1)))\subseteq\Enab_{S_2}(\alpha_{12}(x_1))\subseteq\Enab_{S_1}(x_1)$.\\ %, where $x_2\leq_{X_2}\alpha_{12}(x_1)$ as $\alpha_{12}$ and $\gamma_{21}$ form a Galois connection \AKS{Check lattices}.\\
  \item (A2.2) As $\alpha_{12}(F_1(x_1,u))\subseteq F_2(\alpha_{12}(x_1),u)$ and  $\alpha_{23}(F_2(x_2,u))\subseteq F_3(\alpha_{23}(x_2),u)$ for any $x_1\in X_1$ and $x_2\in X_2$ it follows that  
  $\alpha_{23}(\alpha_{12}(F_1(x_1,u))
  \subseteq\alpha_{23}(F_2(\alpha_{12}(x_1),u))
%   =\alpha_{23}(F_2(x_2,u))
%   \subseteq(F_3(\alpha_{23}(x_2),u)) 
  \subseteq(F_3(\alpha_{23}(\alpha_{12}(x_1)),u))
  $.\\
\item (A3) As $\Omapn{H_1(\gamma_{21}(x_2)}{1}\subseteq \Omapn{H_2(x_2)}{2}$ and $\Omapn{H_2(\gamma_{32}(x_3)}{2}\subseteq \Omapn{H_3(x_3)}{3}$ for any $x_2\in X_2, x_3\in X_3$ it follows that
$\Omapn{H_1(\gamma_{21}(\gamma_{32}(x_3))}{1}\subseteq\set{\Omapn{H_2(\gamma_{32}(x_3))}{2}}\subseteq\Omapn{H_3(x_3)}{3}$.
 \end{inparaitem}
\end{proof}

The following technicallemmas are used in the analysis of the KAM algorithm.

\begin{lemma}\label{lem:singlecover}
 After execution of the function \textsc{Refine} in line~\ref{line:GSSA:CallRefine} of \REFalg{alg:GSSA}, it holds that $|\alphat(\xt)|=1$ for all $\xt$ for which $\alphat$ is defined.
\end{lemma}

\begin{proof}
First observe that \textsc{Refine}  is only called if $\xt\subset \xa$. If $\xt=\xa$ the claim is trivially satisfied as $\xa$ is the unique minimal element covering $\xt$ in this case. 
As $Q'$ in line~\ref{line:GSSA:Qp} of \REFalg{alg:GSSA} is chosen to be minimal, we have that $\tuplel{\xa_1,\xt},\tuplel{\xa_2,\xt}\in\ExG$ with $\xa_1\neq\xa_2$ implies $\xa_1\not\subseteq\xa_2$ and $\xa_2\not\subseteq\xa_1$ and in addition $\xt\subset\xa_1$ and $\xt\subset\xa_2$, so \textsc{Refine} is called. Further, as \REFalg{alg:GSSA} is initialized with a cover which partitions the state space, we know that there exists a minimal $\xa$ which was split into $\xa_1\subset \xa$ and $\xa_2\subset \xa$ previously. This implies that there exists $\xt_1$ and $\xt_2$ s.t.\ $F(\xa_1,u)=\mathtt{PostQ}_u(\tuplel{\xa_1,\xt_1})$ and $F(\xa_2,u)=\mathtt{PostQ}_u(\tuplel{\xa_2,\xt_2})$ while there exists some $x_1\in\xa_1\setminus \xa_2$ s.t.\ $x\notin \mathtt{PostQ}_u(\tuplel{\xa_2,\xt_2})$ and, vise versa, there exists some  $x_2\in\xa_2\setminus \xa_1$ s.t.\ $x\notin \mathtt{PostQ}_u(\tuplel{\xa_1,\xt_1})$, as otherwise the cover cell $\xa$ would not have been splitted. Now, consider $\xt$ from 
before, and observe that $\xt\in\xa_1\cap\xa_2$ by definition. Further, the above reasoning implies $\xt\subset\xt_1$ and $\xt\subset\xt_2$, and $\mathtt{PostQ}_u(\tuplel{\xa_1,\xt})\subset\mathtt{PostQ}_u(\tuplel{\xa_1,\xt_1})=F(\xa_1,u)$ and $\mathtt{PostQ}_u(\tuplel{\xa_2,\xt})\subset\mathtt{PostQ}_u(\tuplel{\xa_2,\xt_2})=F(\xa_2,u)$ with proper containment in both cases. 
This introduces a contradiction to the assumption that \textsc{Refine} has terminated, as in this case we know that $\xa_1$ and $\xa_2$ cannot be further splitted, i.e., $F(\xa_1,u)=\mathtt{PostQ}_u(\tuplel{\xa_1,\xt})$ and $F(\xa_2,u)=\mathtt{PostQ}_u(\tuplel{\xa_2,\xt})$. The last equality only holds if $\xa_1=\xa_2$ as in this case no $x_1$ and $x_2$ as above can be constructed. 
\end{proof}

% \begin{lemma}\label{lem:TerminatingGSSAifXo}
% Let $S$ be a system s.t.\ $\Xo=X$. Further, let $\Cover_l$ be the cover set of \REFalg{alg:GSSA} whenever line~\ref{line:ComputeAbs:Se} is reached in the $l$-th iteration of the while loop in line~\ref{alg:GSSA:startwhile}-\ref{alg:GSSA:endwhile}. If $\Cover_l=\Cover_{l-1}$ for $l>1$, then $\Cover$ is stable. 
% \end{lemma}
% 
% \begin{proof}
%  Consider the $l$-th iteration of \REFalg{alg:GSSA} and observe that $\Cover_l=\Cover_{l-1}$ implies that no cell $c$ contained in $\ExG$ is refined by \textsc{Refine}, that is for all $\tuplel{c,q}\in\ExG$ holds that $F(q,u)=\ON{PostQ_u(c)}$ and $c$ has a unique cover block $q$ (from ...). Further, we know that for all $x\in X$ there exists a $c$ s.t.\ $x\in c$, as $\Xo=X$.
%  Now consider a new $c'\subset c$ s.t.\ $c'=F(c,u)\cap H^{-1}(y)\neq\emptyset$ and recall that $\tuplel{c',q}\in\ExX^\downarrow$. 
% \end{proof}

\begin{lemma}\label{lem:TerminatingGSSAifSSA}
 If \REFalg{alg:SSA} terminates, there exists an iteration $l\in\mathbb{N}$ of \REFalg{alg:GSSA} for which $\ExG=\ExX^\downarrow$ holds.
\end{lemma}

\begin{proof}
First, it can be verified that in every iteration of the while loops in 
both algorithms $\tuplel{\cdot,c}$ gets added to $\ExX^\downarrow$ in \REFalg{alg:GSSA} iff $c$ gets added to $\Xa$ in \REFalg{alg:SSA}. 
This is due to the fact that the set of minimal blocks covering $c$ is uniquely defined and refinements of any block 
are propagated through all sets $\ON{EXP}_{\Gamma,X,F}$ within \textsc{Refine}. 
Therefore, it cannot happen that a tuple $\tuplel{q,c}$ is added to $\ExX^\downarrow$ if $\ExX^\downarrow$ already contains a tuple $\tuplel{q',c}$. Thus, the termination conditions of the while loops coincide.

As \textsc{Refine} is a recursive function, we have to additionally prove that it terminates. To see this, observe that $\ExF$ is a finite tree for every initial tuple $\tuplel{\varepsilon,q,q}$ and therefore only contains finite paths. Further, as every iteration of the while loop in line~\ref{alg:GSSA:startwhile}-\ref{alg:GSSA:endwhile} of \REFalg{alg:GSSA} only explores the current leaves of this tree, it adds new leaves to the tree and schedules leaves of the previous iteration for possible refinement. As the recursion of \textsc{Refine} in line~\ref{alg:GSSA:recurseRefine} of \REFalg{alg:GSSA} only schedules predecessors of these leaves and the tree is finite, it terminates.
\end{proof}

\end{appendix}

\end{document}